\documentclass[12pt,letterpaper,oneside,english]{amsart}

\usepackage{xr}

\usepackage[T1]{fontenc}
\usepackage[latin9]{inputenc}
\usepackage{xcolor}
\usepackage{babel}
\usepackage{array}
\usepackage{float}
\usepackage{booktabs}
\usepackage{url}
\usepackage{multirow}
\usepackage{tablefootnote}
\usepackage{amstext}
\usepackage{amsthm}
\usepackage{amssymb}
\usepackage{graphicx}
\usepackage{setspace}
\usepackage[authoryear]{natbib}
\usepackage[unicode=true]
 {hyperref}

\graphicspath{{../}}
\makeatletter

\providecommand{\tabularnewline}{\\}

\theoremstyle{plain}
\newtheorem{thm}{\protect\theoremname}

\usepackage{lmodern}
\usepackage[T1]{fontenc}
\usepackage{booktabs}
\usepackage{hyperref}
\hypersetup{
hidelinks
}

\setcitestyle{authoryear,open={(},close={)}}
\theoremstyle{definition}
\newtheorem{pred}{Prediction}
\theoremstyle{theorem}
\newtheorem{result}{Result}

\usepackage[caption=false]{subfig}
\makeatother

\providecommand{\theoremname}{Theorem}

\begin{document}
\title{Revealing Choice Bracketing}
\author{Andrew Ellis and David J. Freeman}
\date{March 2024}
\thanks{Freeman thanks UCSB and QMUL for hosting him while parts of this paper
were written, and UToronto for use of its TEEL lab for our experiments.
We especially thank Johannes Hoelzemann for help in coordinating experiments
at TEEL. We thank Pietro Ortoleva, Stefano DellaVigna, three anonymous
referees, Ted Bergstrom, Ignacio Esponda, Erik Eyster, Daniel Gottlieb,
Yoram Halevy, Alex Imas, Matt Levy, Marco Mariotti, Antony Millner,
Ryan Oprea, Luba Petersen, Kate Smith, Balazs Szentes, Emanuel Vespa,
Sevgi Yuksel, Lanny Zrill, and audiences at UCSB, Berlin, Chapman,
QMUL, UCL, UTS, Berkeley/Booth, UCSD, Ottawa, UC Davis, D-TEA 2020,
ESA 2020, RUD 2020, ESNAWM 2021, MBEES 2021, CEA 2021, BRIC 2022 for
helpful discussions and suggestions. We thank Alex Ballyk, Daphne
Baldassari, Matheus Thompson Bandeira, Priscilla Fisher, En Hua Hu,
Louise Song, Johnathan Wang, and Marie Zagre for research assistance.
This paper was funded by a 2018 SFU/SSHRC Institutional Grant and
SSHRC Insight Grant 435-2019-0658 and was conducted under SFU Research
Ethics Study \#2016s0380.}
\begin{abstract}
Experiments suggest that people fail to take into account interdependencies
between their choices -- they do not broadly bracket. Researchers
often instead assume that people narrowly bracket, but existing designs
do not test it. We design a novel experiment and revealed preference
tests for how someone brackets their choices. In portfolio allocation
under risk, social allocation, and induced-value shopping experiments,
40-43\% of subjects are consistent with narrow bracketing and 0-16\%
with broad bracketing. Adjusting for each model's predictive precision,
74\% of subjects are best described by narrow bracketing, 13\% by
broad bracketing, and 6\% by intermediate cases.


\end{abstract}

\thanks{Ellis: Department of Economics, London School of Economics. e-mail:
\href{mailto:a.ellis@lse.ac.uk}{a.ellis@lse.ac.uk}. Web: \url{http://personal.lse.ac.uk/ellisa1/}}
\thanks{Freeman: Department of Economics, Simon Fraser University. e-mail:
\href{mailto:david_freeman@sfu.ca}{david\_freeman@sfu.ca}. Web: \url{http://www.sfu.ca/~dfa19/}}

\maketitle
Individuals face many interconnected decisions. How an individual
takes into account the interdependencies when choosing, or how they
\emph{bracket} these choices, significantly influences their decision-making
process. Bracketing determines which outcomes are evaluated as gains
or losses and fair or unfair, and also plays a role in measuring parameters
like risk aversion. Nearly every behavioral model and most ``rational''
ones require some assumption about how people bracket choices.

There are many ways to bracket. Optimal decision-making requires that
people \emph{broadly bracket}, considering every feasible combination
of choices and selecting the best. The most common alternative is
that people \emph{narrowly bracket} by making each decision without
considering any interdependencies. However, these two extremes are
far from exhaustive. For instance, \citet{barberis2006individual}
and \citet{rabin2009narrow} propose a hybrid of the two called \emph{partial-narrow
bracketing}.

Most experimental evidence interpreted as being \emph{for} narrow
bracketing is actually evidence \emph{against} broad bracketing. This
evidence, surveyed in Section \ref{sec:Evidence}, comes mainly from
studies that follow a similar design to \citet[Problem 3]{tversky1981framing}
or \citet[Problems 11-12]{kahneman1979}. In the former, each subject
makes two concurrent choices, and one pair of choices generates a
distribution over outcomes dominated by another pair. They find that
many subjects choose the dominated pair. In the latter, two groups
of subjects face choices between lotteries that are economically identical
but differ in how payments are divided between an endowed income and
an active choice. The two groups make different choices. Both designs
provide evidence against broad bracketing. However, narrow bracketing
makes no testable predictions in either design: any choices are consistent
with it. This leaves open the question of whether narrow bracketing
is a good description of behavior.

We propose a theoretical framework and experimental design to test
how an individual brackets. Every subject makes several decisions,
each comprised of one or more parts. Their payoff from a decision
is determined by the sum of the items chosen in its parts. A subject
is consistent with narrow bracketing if they maximize a preference
relation in each part, and consistent with broad bracketing if they
optimize after integrating all parts of the decision into a single
feasible set. We characterize all the testable predictions for narrow,
broad, and partial-narrow bracketing. For example, we check for narrow
bracketing by performing a standard revealed preference test on a
dataset of their choices that treats each part as an independent observation.
The results provide individual-level, non-parametric tests that rely
only on monotonicity of the underlying preferences. They provide the
theoretical basis for an experimental design that tests all three
models of choice bracketing.

To implement these tests, we conduct experiments in which subjects
make five decisions consisting of one or two parts. Every part is
a budget set, and integrating the choices only requires addition.
Unlike the existing designs described above, narrow and broad bracketing
make distinct predictions in each of the two-part decisions. Every
good in the first part is a perfect substitute for a good in the second.
Broad bracketers recognize and take full advantage of this substitutability,
while narrow bracketers do not.

We apply our design and tests to portfolio choices among risky assets
\citep{choi2007consistency}, dividing money between two anonymous
other subjects \citep{andreoni2002giving,fisman2007individual}, and
standard consumer problems with induced values over bundles. The last
allows us to conduct even more powerful tests of bracketing because
preferences are known. Across the three experiments, 40-43\% of subjects
are consistent with narrow bracketing, 0-16\% are consistent with
broad, 1-32\% are consistent with partial-narrow but neither broad
nor narrow, and 24\%-44\% are consistent with none of the three. No
subjects are consistent with \emph{both} broad and narrow bracketing.
We then determine which model best describes each subject's bracketing
based on both the number of errors they make relative to each model
and the predictive precision of that model \citep{selten1991properties}.
We classify 6-13\% of subjects to broad bracketing, 68-78\% to narrow
bracketing, and only 3-9\% to partial narrow bracketing, with variation
across experiments.\footnote{The remaining subjects are poorly described by all models and left
unclassified.}

Failing to bracket broadly has consequences for welfare. Economists
as early as \citeauthor{smith1776inquiry} (1776) argue that agents
benefit from specialization. A broad bracketer specializes by purchasing
as much of a good as possible when it is relatively cheap before purchasing
any when it has a higher opportunity cost. This increases utility
in the same way that comparative advantage increases total productivity
in classic trade models.\footnote{\citet{baron2004support} provide a survey measure of understanding
of specialization, and show it is correlated with attitudes to free
trade.} Broad bracketers apply Smith's insights about the productivity increases
from specialization to individual decision-making, while narrow bracketers
fail to reap the gains from specialization.\footnote{We thank a referee for pointing to this analogy.}
Our approach allows us to quantify the magnitude of these losses.
For instance, the subjects classified as narrow bracketers made an
average of \$1.29 less than the broad bracketers from the two-part
decisions in the Shopping Experiment, more than 10\% of the variable
payment.

Because our tests are individual-level, they can provide evidence
on why so many people bracket narrowly despite the gains from broad
bracketing. In online follow-up experiments, we recorded choice process
data and implemented a nudge to encourage subjects to ``examine both
parts'' of the decision. While the nudge increased the proportion
of subjects who looked at both parts, it had a limited effect on the
rate of broad bracketing. In one arm of the Online Risk Experiment,
we observed that about a quarter of the narrow bracketers had enough
information to bracket broadly but did not do so. In the Online Shopping
Experiment, a similar fraction used a calculator to compute the payoff
of one or more bundles that were only feasible at the decision level
and not in any part on its own. Both these behaviors indicate some
consideration of how choices from both parts combine to determine
payoffs. This suggests that not all narrow bracketing results from
simply looking at one part of a decision at a time and making an optimal
choice for that part on its own. Instead, it suggests that some individuals
consider both decision problems, yet ultimately implement choices
that are only compatible with narrow bracketing.

Bracketing also affects how one interprets behavior. A narrow bracketer's
reluctance to take a small yet actuarially favorable gamble indicates
only slight risk aversion. However, for broad bracketers, the rejection
of the same gamble signifies extreme risk aversion, given its minimal
impact on their overall wealth \citep{rabin2000risk}. Our approach
enables us to separate underlying preferences from bracketing, with
implications for measuring risk or inequity aversion.

This separation allows us to measure heterogeneity in bracketing directly
and to account for it when inferring preferences. To illustrate the
importance of doing so, consider the Social Experiment. In Part 2
of Decision 1 (see Table \ref{tab:design}), subjects on average allocate
unequally between the two anonymous others, dividing \$16 into \$6.80
and \$9.20. If all subjects bracketed narrowly, this would suggest
a lack of concern for equity. However, we classify 10\% of subjects
as broad bracketers and 75\% as narrow bracketers. The narrow bracketers
on average allocate \$7.75 and \$8.25 to each person, and the broad
bracketers all achieve perfectly equal allocations in the decision
by allocating a \$14 - \$2 split in that part and a \$0 - \$12 split
in the other. After taking bracketing into account, all these subjects
are consistent with inequity aversion. The only existing work that
estimates heterogeneity in bracketing, \citet{rabin2009narrow}, estimates
a between-subjects, parametric, structural model and assumes that
all subjects have the same preferences. In contrast, our approach
is individual-level, non-parametric, and allows heterogeneity in preferences.

Most applications of prospect theory explicitly assume narrow bracketing
(\citealt{camerer2004prospect}, Table 5.1; \citeauthor{o2018reference},
2018, Sections 3.4 and 4.5). Since it is impossible to avoid all other
risks, a broad bracketer will not exhibit noticeable small-stakes
loss aversion \citep{barberis2006individual}. At the same time, other
applications of prospect theory, such as casino gambling \citep{barberis2012model,ebert2015until},
require broader bracketing. Our results suggesting heterogeneity in
bracketing may reconcile the results rejecting broad bracketing (e.g.
\citealp{tversky1981framing}) and those more supportive of it (e.g.
\citealp{heimer2020dynamic} and \citealp{baillon2022randomize}).

Choice bracketing is also important for social decisions. For example,
\citeauthor{sobel2005interdependent} (\citeyear[p. 400]{sobel2005interdependent})
remarks that in experiments studying social decisions, a subject should
``maximize her monetary payoff in the laboratory and then redistribute
her earnings to deserving people later... if the concern for inequity
was `broadly bracketed.'{}'' However, the literature he surveys and
our analysis both find evidence consistent with subjects caring about
narrowly-bracketed equity (e.g. \citealp{charness2002understanding,bolton2000erc}).

Standard economic analyses make assumptions about bracketing as well.
As noted by a recent behavioral economics textbook, ``Although choice
bracketing has been largely ignored in economics, models in economic
theory often implicitly assume and invoke it'' (\citealp{dhami2016foundations},
p. 1469). Many analyses exclude some related choices from the model
yet interpret results in terms of a fully rational choice. Our results
suggesting that most people are narrow bracketers lend support to
this approach, with the caveat that some seem best modeled as broad
bracketers.

\section{Testing choice bracketing\label{sec:Theory}}

We consider a decision-maker (DM) who faces $T\geq1$ decision problems
involving alternatives contained in $\mathbb{R}_{+}^{n}$. Decision
$t$ consists of $K_{t}\geq1$ parts. Formally, the feasible set of
part $k$ from decision $t$ is $B^{t,k}$, assumed to be compact
and non-empty. The DM chooses one alternative, denoted $x^{t,k}$,
from each part, $B^{t,k}$. Thus, the analyst observes the \emph{dataset}
\[
\mathcal{D}=\left\{ \left(x^{t,k},B^{t,k}\right)\right\} _{(t,k)}
\]
indexed by decisions and parts and where $x^{t,k}\in B^{t,k}$. They
face parts concurrently, and their payoff in decision $t$ depends
only on the sum of their choices across parts of the decision, 
\[
x^{t}=\sum_{k=1}^{K_{t}}x^{t,k},
\]
which we call the ``final alternative.'' A real world analogue consists
of a scanner dataset with purchases at different stores (parts, indexed
by $k$) in a given time period (a decision, indexed by $t$), and
the shopper consumes these purchases after buying at all stores but
before the next period.

A DM \emph{broadly brackets} if they maximize an increasing utility
function over the set of feasible final alternatives for each decision
problem,
\[
B^{t}=\left\{ \sum_{k=1}^{K_{t}}y^{t,k}:y^{t,k}\in B^{t,k}\right\} .
\]
In contrast, they \emph{narrowly bracket} if they choose the alternative
in part $k$ of decision $t$ that maximizes an increasing utility
function over $B^{t,k}$.\footnote{A narrow bracketer acts as if they perceive alternatives correctly
and maximizes a well-behaved preference over them, but misperceives
the budget set. In contrast, a DM who misperceived correlation, e.g.
\citet{eyster2016correlation} or \citet{EllisPiccione2017}, perceives
the choice set correctly but misperceives the alternatives themselves.} Both require that the DM is ``rational'' in that they maximize
some well-behaved preference relation, but neither requires any further
assumptions about preferences beyond monotonicity. A narrow bracketer
optimizes part-by-part, while a broad bracketer does so decision-by-decision.

Our goal is to test which, if any, models of bracketing can explain
a subject's choices. Formally, a dataset $\mathcal{D}$ is \emph{rationalized
by broad bracketing} if there exists an increasing utility function
$u:\mathbb{R}_{+}^{n}\rightarrow\mathbb{R}$ so that
\[
x^{t}=\arg\max_{x\in B^{t}}u\left(x\right)
\]
for every decision $t$, and a dataset $\mathcal{D}$ is \emph{rationalized
by narrow bracketing} if there exists an increasing $u:\mathbb{R}_{+}^{n}\rightarrow\mathbb{R}$
so that 
\[
x^{t,k}=\arg\max_{x\in B^{t,k}}u\left(x\right)
\]
 for every part $k$ and decision $t$.\footnote{These definitions assume that each choice represents strict preference,
but one can easily generalize them and the results that follow to
allow for indifference.} The next two subsections provide necessary and sufficient conditions
for $\mathcal{D}$ to be rationalized by either of the two.

\subsection{Predictions}

We first provide necessary conditions for rationalization by broad
and narrow bracketing. The first two predictions build on the observation
that the DM's choice reveals their preference among the set of alternatives
over which they optimize. The predictions combine the Weak Axiom of
Revealed Preference (WARP) with a form of bracketing.

\begin{pred}[NB-WARP]\label{pred: first}Suppose that $\mathcal{D}$
is rationalized by narrow bracketing.\\
 If $\left(x^{t,k},B^{t,k}\right),\left(x^{t^{\prime},k^{\prime}},B^{t^{\prime},k^{\prime}}\right)\in\mathcal{D}$
and $x^{t^{\prime},k^{\prime}}\in B^{t,k}\subseteq B^{t^{\prime},k^{\prime}}$,
then $x^{t,k}=x^{t^{\prime},k^{\prime}}$.

\end{pred}

\begin{pred}[BB-WARP]Suppose that $\mathcal{D}$ is rationalized
by broad bracketing.\\
 For decisions $t,t^{\prime}$ in $\mathcal{D}$, if $x^{t^{\prime}}\in B^{t}\subseteq B^{t^{\prime}}$,
then $x^{t}=x^{t^{\prime}}$.

\end{pred}

Narrow bracketing requires that WARP holds when comparing any pair
of parts of decisions, even when they belong to economically different
decisions. Broad bracketing implies that WARP holds at the decision
level, comparing final alternatives that are feasible in both aggregate
budget sets.

The next prediction reflects the appropriate manifestation of monotonicity
in our setting.

\begin{pred}[BB-Mon]\label{pred:last}Suppose that $\mathcal{D}$
is rationalized by broad bracketing.\\
For any decision $t$ in $\mathcal{D}$ and any $y\in B^{t}$, if
$y\geq x^{t}$, then $y=x^{t}$.

\end{pred}

BB-Mon requires that the subject chooses on the frontier of their
aggregate budget set in a given decision.\footnote{Narrow bracketing predicts a similar condition at the part level.
Our experimental implementation forces it to hold, so we do not formally
include it here.} 

\subsection{Characterization}

The above predictions are necessary but not sufficient conditions
for each type of bracketing. We obtain tight characterizations by
extending the logic of the above to include indirect implications,
in the same manner that Strong Axiom of Revealed Preference (SARP)
extends WARP. Theorem \ref{thm:The-following-are} shows that these
are necessary and sufficient conditions for a given dataset to be
rationalized by a particular form of bracketing.

Our tests are based on applying SARP to an ancillary dataset. Say
that a bundle $x$ is \emph{directly revealed preferred} \emph{to
}$y$ in a dataset $\mathcal{D}^{\prime}$ consisting of single-part
decisions, written $xP^{\mathcal{D}^{\prime}}y$, if either $x\geq y$
and $x\neq y$ or there exists $\left(x,B\right)\in\mathcal{\mathcal{D}^{\prime}}$
so that $y\in B\backslash\{x\}$. A dataset $\mathcal{D}^{\prime}$
\emph{satisfies SARP} if the binary relation $P^{\mathcal{\mathcal{D}^{\prime}}}$
is acyclic.

We say that say that $\mathcal{D}$ satisfies BB-SARP if the ancillary
dataset 
\[
\mathcal{D}^{BB}=\left\{ \left(x^{t},B^{t}\right)\right\} _{t\in T}.
\]
satisfies SARP. Similarly, we say that say that $\mathcal{D}$ satisfies
NB-SARP if SARP is satisfied by the ancillary dataset 
\[
\mathcal{D}^{NB}=\left\{ \left(x^{s},B^{s}\right)\right\} _{s=1}^{\sum_{t=1}^{T}K_{t}}
\]
where for each part $(t,k)$, there exists a unique $s$ so that
$\left(x^{s},B^{s}\right)=\left(x^{t,k},B^{t,k}\right)$. In $\mathcal{D}^{NB}$,
each part of every decision is treated as a separate, independent
observation, whereas each decision is in $\mathcal{D}^{BB}$.
\begin{thm}
\label{thm:The-following-are}The following are true:\\
(i) The dataset $\mathcal{D}$ satisfies BB-SARP if and only if $\mathcal{D}$
is rationalizable by broad bracketing, and\\
(ii) The dataset $\mathcal{D}$ satisfies NB-SARP if and only if $\mathcal{D}$
is rationalizable by narrow bracketing.
\end{thm}
Theorem \ref{thm:The-following-are} characterizes the complete testable
implications of broad and narrow bracketing. Moreover, both conditions
are readily applied using standard computational tools.\footnote{We note that $B^{t}$ typically has a piece-wise linear budget-line
even if each $B^{t,k}$ is a standard budget set.} One can make tighter predictions by imposing more structure on the
utility function that rationalizes the data. For instance, we impose
that preferences are symmetric when we apply these tests to our experiments.
This enables stronger versions of the two tests.\footnote{Specifically in $\mathbb{R}_{+}^{2}$, if $\left(x,y\right)P^{\mathcal{D}}\left(x^{\prime},y^{\prime}\right)$,
then $\left(y,x\right)P^{\mathcal{D}}\left(x^{\prime},y^{\prime}\right)$,
$\left(x,y\right)P^{\mathcal{D}}\left(y^{\prime},x^{\prime}\right)$,
and $\left(y,x\right)P^{\mathcal{D}}\left(y^{\prime},x^{\prime}\right)$.
In Appendix \ref{sec:Visualizations}, we derive some immediate implications
of symmetry.}

\subsection{Partial-narrow bracketing}

We also consider an intermediate model of bracketing lying between
the two extremes, called \emph{$\alpha$-partial-narrow bracketing
(PNB). }Inspired by \citet{barberis2006individual}, \citet{barberis2009preferences},
and \citet{rabin2009narrow}, the dataset $\mathcal{D}$ is \emph{rationalized
by $\alpha$-partial-narrow bracketing }if
\begin{equation}
\left(x^{t,1},\dots,x^{t,K_{t}}\right)=\arg\max_{y^{t,k}\in B^{t,k}\ \forall k}\alpha\sum_{k=1}^{K^{t}}u\left(y^{t,k}\right)+\left(1-\alpha\right)u\left(\sum_{k=1}^{K^{t}}y^{t,k}\right)\label{eq:partial}
\end{equation}
for all $t$.\footnote{The first two papers consider an average of the certainty equivalents,
{[}$u^{-1}(E[u(x^{t})])+b_{0}\sum_{k^{\prime}=1}^{K^{t}}v^{-1}(E[v(x^{t,k})])$.{]}
Unlike our specification, they allow for the narrow utility function
to differ from the broad utility function and assume expected utility.
The algorithm in Theorem \ref{thm: alpha PNB} can be adapted to allow
for different narrow and broad utility functions ($u\neq v$) at the
cost of less predictive power. Our specification is closest to \citet[p. 1513]{rabin2009narrow},
although they also assume expected utility.} A PNB DM's choices maximize a weighted average of the utility of
the final alternative and of the utility of their choice in each part.

A broad bracketer takes into account that good $i$ in part $k$ is
a perfect substitute for good $i$ in part $k^{\prime}$, whereas
a narrow bracketer does not consider any complementarities across
parts of the decision.\footnote{\citet{koszegi2020choice} also base their approach to mental budgeting
on substitutability between parts.} In contrast, a PNB DM views the same good in different parts as imperfect
substitutes for each other. They act as if they make a single choice
in decision $t$ of an element of $\mathbb{R}_{+}^{nK_{t}}$, but
they maximize a utility function derived from $u$ rather than $u$
itself. In contrast, a broad bracketer acts as if they choose a single
element of $\mathbb{R}_{+}^{n}$, and a narrow bracketer as if they
make $K_{t}$ independent choices of elements of $\mathbb{R}_{+}^{n}$.\footnote{Recent models by \citet{vorjohann2020reference} and \citet{zhang2023procedural}
have this feature as well, but with different utilities over $\mathbb{R}_{+}^{nK_{t}}$
than Equation (\ref{eq:partial}). They capture departures from broad
bracketing by either substituting in a reference point for other parts
or by imposing separability where $u$ is inseparable. In a two part
decisions, these models are as follows. \citet{vorjohann2020reference}
considers a DM who maximizes $v\left(x^{t,1},x^{t,2}\right)=u\left(x^{t,1},r^{2}\right)+u\left(r^{1},x^{t,2}\right)$
for reference points $r^{1}\in B^{t,1}$ and $r^{2}\in B^{t,2}$,
instead of $v^{BB}\left(x^{t,1},x^{t,2}\right)=u\left(x^{t,1},x^{t,2}\right)$.
\citet{zhang2023procedural}, for the specification most similar to
our setting, considers a DM who maximizes $U\left(x^{t,1},x^{t,2}\right)=E_{x^{t,1}}\left[u\left(z+u^{-1}\left(E_{x^{t,2}}\left[u(y)\right]\right)\right)\right]$
instead of $U^{BB}\left(x^{t,1},x^{t,2}\right)=E_{x^{t,1}+x^{t,2}}\left[u\left(z\right)\right]$,
where $E_{x}$ is the expectation operator with respect to the lottery
corresponding to $x$.}

An algorithm \emph{tests }whether $\mathcal{D}$ is rationalizable
by $\alpha$-partial-narrow bracketing if it inputs a dataset $\mathcal{D}$
and (deterministically) outputs 1 when $\mathcal{D}$ is rationalizable
by $\alpha$-partial-narrow bracketing and outputs 0 otherwise.
\begin{thm}
\label{thm: alpha PNB}There exists an algorithm that tests whether
$\mathcal{D}$ is rationalizable by $\alpha$-partial-narrow bracketing.
\end{thm}
We construct the algorithm in Appendix \ref{sec:Theoretical-Appendix}.
The algorithm reduces the problem of determining whether $\mathcal{D}$
is consistent with $\alpha$-partial-narrow bracketing to that of
a standard linear programming problem. It runs in polynomial time
with a fixed set of possible outcomes of each decision and its parts.
We elaborate on the ideas behind it, and provide a formal proof in
the Appendix.

The algorithm takes advantage of the formal similarity between Equation
(\ref{eq:partial}) and expected utility. Specifically, an $\alpha-$PNB
DM with utility $u$ chooses $x^{t,1}$ in part 1 and $x^{t,2}$ in
part 2 if and only if an expected utility DM with Bernoulli index
$u$ chooses $\left(\psi\alpha,x^{t,1};\psi\alpha,x^{t,2};\psi(1-\alpha),x^{t}\right)$
from the menu of lotteries $\left\{ \left(\psi\alpha,y^{1};\psi\alpha,y^{2};\psi(1-\alpha),y^{1}+y^{2}\right):y^{1}\in B^{t,1}\ \&\ y^{2}\in B^{t,2}\right\} $
where $\psi=(1+\alpha)^{-1}$. The algorithm first transforms the
original dataset to an ancillary dataset where each decision is a
single choice from a menu of lotteries. Each lottery in the menu is
structured so that an alternative in part $k$ occurs with probability
proportional to $\alpha$ for each part $k$ and their sum occurs
with probability proportional to $(1-\alpha)$. The DM's choice is
the lottery over their choices in each part and their sum. Existing
results, e.g. \citep{clark1993revealed}, provide necessary and sufficient
conditions for the ancillary dataset to be rationalizable by expected
utility. We show that whether or not it can be rationalized can be
determined by examining the solution to a linear program.

We also consider a related intermediate model of bracketing, ``PNB-Personal
Equilibrium (PNB-PE).'' Formally,
\[
x^{t,k}=\underset{y^{t,k}\in B^{t,k}}{\arg\max}\left[\alpha u\left(y^{t,k}\right)+(1-\alpha)u\left(\underset{k^{\prime}\in\left\{ 1\leq k^{\prime}\leq K^{t}:k^{\prime}\neq k\right\} }{\sum}x^{t,k^{\prime}}+y^{t,k}\right)\right].
\]
for every part $k$ of decision $t$. While the functional form is
similar to PNB, the DM acts as if they make $K_{t}$ independent choices
of elements of $\mathbb{R}_{+}^{n}$, as a narrow bracketer would
would. Their choices are the personal equilibrium (PE) of an intrapersonal
game \citep{koszegi2006model,lian2019theory}. A PNB-PE DM only considers
the utility effect of $x^{t,-k}$ when choosing $x^{t,k}$, ignoring
that they could change $x^{t,-k}$ in response to their choice. A
similar algorithm to the one that tests $\alpha$-PNB can test $\alpha$-PNB-PE.

\section{Existing Evidence on Bracketing\label{sec:Evidence}}

The models we consider capture different forms of what \citet{read1999choice}
call choice bracketing: how a person groups ``individual choices
together into sets.'' We first discuss three strands of research
into static, concurrent decisions to which our setting directly applies
before turning to the important case of dynamic, or temporal, bracketing.

Most existing direct evidence of non-broad choice bracketing is derived
from \citet{tversky1981framing}, in which subjects face the decision
problem described in Table \ref{table: TK decision}.
\begin{table}[h]
\caption{Tversky and Kahneman's Decision problem}
\label{table: TK decision}
\begin{quotation}
\noindent Imagine that you face the following pair of concurrent decisions.
First, examine both decisions, then indicate the options you prefer.\\
\begin{centering}
\textbf{Decision (i). Choose between:}
\par\end{centering}
\begin{minipage}[t]{0.4\columnwidth}%
A. a sure gain of \$240 \textcolor{darkgray}{\small{}{[}85\%{]}}%
\end{minipage}%
\begin{minipage}[t]{0.5\columnwidth}%
B. 25\% chance to gain \$1000,

and 75\% chance to gain nothing \textcolor{darkgray}{\small{}{[}16\%{]}}{\small\par}%
\end{minipage}

\smallskip{}
\begin{centering}
\textbf{Decision (ii). Choose between:}
\par\end{centering}
\begin{minipage}[t]{0.4\columnwidth}%
C. a sure loss of \$750 \textcolor{darkgray}{\small{}{[}13\%{]}}%
\end{minipage}%
\begin{minipage}[t]{0.5\columnwidth}%
D. 75\% chance to lose \$1000,

and 25\% chance to lose nothing \textcolor{darkgray}{\small{}{[}87\%{]}}{\small\par}%
\end{minipage}
\end{quotation}
\end{table}
 This design fits into our theoretical setting as a single-decision
dataset with $B^{1,1}=\{A,B\}$ and $B^{1,2}=\{C,D\}$. The choice
combination $(A,D)$ made by 73\% of their subjects generates a first-order
stochastically-dominated distribution over outcomes compared to the
pair of choices $(B,C)$ and so violates BB-Mon and BB-SARP. Their
design cannot falsify narrow bracketing since every combination of
choices satisfies NB-SARP. Without further restrictions, their results
are uninformative about the choice bracketing of subjects who do not
choose A and D. Yet about 70\% of subjects made non $(A,D)$ choices
in incentivized follow-up experiments \citep{rabin2009narrow,koch2019correlates}.

A related set of experiments considers how an exogenously fixed quantity,
like a monetary endowment, an asset, or a background risk, affects
a person's choice when separated from the description of available
alternatives.\footnote{Related theoretical results by \citet{samuelson1963risk}, \citet{rabin2000risk},
\citet{safra2008calibration}, and \citet{mu2020background} suggest
that behavior in small-stakes bets cannot be reconciled with behavior
in large-stakes bets and broad bracketing. The results of \citet{gneezy1997experiment,benartzi1995myopic}
provide support for their theoretical predictions.} For instance, studies suggest subjects behave as-if they do not incorporate
their endowment in risk-taking choices \citep[Problems 11-12]{kahneman1979},
in social allocation tasks \citep{exley2018equity}, and in labor
supply choices \citep{fallucchi2021narrow}. These designs fit into
our theoretical framework as decisions in which one part is a singleton
set (the endowment) and one part a non-singleton set (the active choice),
and these papers provide evidence against broad bracketing.

The final body of evidence shows how failures of fungibility affect
choice \citep{thaler1999mental}. Studies of the ``flypaper effect''
suggest that transfers earmarked for a particular type of spending
tend to actually be spent there \citep{hines1995flypaper,abeler2017fungibility}.
Empirical evidence from consumption choices after unexpected price
changes support the lack of fungibility \citep{hastings2013fungibility,hastings2018snap}.\footnote{\citet{evers2019mental} suggest that similarity mediates how people
group outcomes into categories.}

This evidence has a more complicated relationship with our setting.
Thaler (\citeyear{thaler1985mental,thaler1999mental}) and others
\citep{galperti2019theory,koszegi2020choice} explain the evidence
through mental accounting (or budgeting). At its most general, this
refers to a person breaking up the overall decision into categories.
In contrast, narrow choice bracketing is a failure to aggregate smaller
parts into a larger decision. These are not mutually exclusive. If
categories and parts coincide, our model of narrow choice bracketing
provides an extreme model of mental accounting. But in general, mental
accounting is consistent with either broadly-bracketing across parts
\citep[p 207]{thaler1985mental} or with narrowly bracketing each
part (e.g. as described by Corollary 1 of \citealp{koszegi2020choice}).\footnote{One can extend our framework by replacing the ``SARP'' part of NB-
or BB- SARP with a revealed preference test for mental accounting,
such as that of \citet{blow2018observable}. Then, we have two joint
tests: one for mental accounting where the budgets for each category
are determined part-by-part, and one for mental accounting where the
budgets for each category are determined at the decision-level.}

Choice bracketing is also relevant in dynamic settings. \citet{gneezy1997experiment}
show that subjects make different choices period-by-period than when
they must choose in advance, evidence of non-broad bracketing.\footnote{The comparison in \citet{cox2015paradoxes} of their ``One Task''
treatment and their ``Pay all sequentially'' treatment provides
between-subjects evidence that fails to reject the null that every
subject narrowly brackets. They note that this is consistent with
prior literature that finds little evidence for ``wealth effects''
(i.e. violations of narrow bracketing) in economics experiments that
pay all choices sequentially.} In contrast, recent evidence from \citet{heimer2020dynamic} (building
on \cite{barberis2012model,ebert2015until}) suggests that subjects behave
as if they do not narrowly bracket future risk-taking opportunities.
In financial decision-making, \citet{shefrin1985disposition} and
\citet{odean1998investors} document the disposition effect, and \citet{thaler1990gambling}
and \citet{imas2016realization} document the house money effect.
Either effect is inconsistent with both fully narrow and fully broad
bracketing.

Our approach can be generalized to a dynamic setting, but  two key
complications arise. First, the setting should comprise decision trees
rather than sets of alternatives. Second, behavior may be dynamically
inconsistent. Our approach can be adapted but would apply only under
dynamic consistency. Future work should incorporate dynamic inconsistency.

We note that there are many ways to fail to bracket broadly in a dynamic
setting. For instance, an agent may only look backwards and take into
account the results of their past decisions but neglect to consider
how their current decision interacts with their future ones. Or they
may only look forwards and take into account future interactions but
not past results. Even more combinations are possible: they could
look neither forwards nor backwards, or only look forwards a certain
number of periods.

\section{Experimental Design\label{sec:Experimental-Design}}

We design and conduct three experiments to test the models of bracketing
in different domains of choice. In each experiment, a participant
faces five decision rounds, each consisting of one or two parts.\footnote{Compared to choice-from-budget-set experiments like \citet{choi2007consistency},
each subject faces fewer rounds of decisions in our experiment and
each budget set is coarser to allow us to conduct the experiment on
paper. We made this design choice to minimize potential for learning
to bracket, and also to slow down subjects' progression through the
experiment to encourage slow and deliberate decision-making.} Each part consists of all feasible integer-valued bundles of two
goods obtained from a linear (or in one case, a piece-wise linear)
budget set. At the end of the experiment, exactly one round is randomly
selected for payment, which we call the ``round that counts''. We
sum all goods purchased in all parts of the decision in the round
that counts to obtain the final bundle that determines payments.\footnote{In the Social Experiment we modify this procedure slightly: one person-decision
pair per anonymous group of two subjects is randomly selected to determine
payments of another anonymous group of two subjects.} By design, there are no complementarities across decisions. We implement
this experimental design to study choice bracketing in three domains
of interest: portfolio choice under risk (Risk), a social allocation
task (Social), and a consumer choice experiment in which we induced
subjects' values (Shopping).

In the Risk Experiment, each part of every decision asks the subject
to choose an integer allocation of tokens between two assets. Each
asset pays off on only one of two equally likely states: Asset A (or
C) pays out only on a die roll of 1-3 whereas Asset B (or D) pays
out only on a die roll of 4-6. The payoff of each asset varies across
decision problems and across parts. Because each decision problem
uses assets with two equally likely states, preferences over portfolios
of monetary payoffs for each state should be symmetric across states.

In the Social Experiment, each part of every decision asks the subject
to choose an integer allocation of tokens between two anonymous other
subjects, Person A and Person B. The value of each token to A and
B varies across decision problems and across parts. Because the two
recipients are anonymous, we expect preferences to be symmetric across
money allocated to A versus B.

In the Shopping Experiment, each part of every decision asks the subject
to choose a bundle of integer quantities of fictitious ``apples''
and ``oranges'' subject to a budget constraint. The monetary payment
for the experiment is calculated from the final bundle in the round
that counts according to the function $\$\text{pay}=\frac{2}{5}\left(\sqrt{\text{\#apples}}+\sqrt{\text{\#oranges}}\right)^{2}$.\footnote{Subjects were provided with a payoff table (Online Supplement \ref{sec:experimental materials},
Figure \ref{tab:payofftable}) to calculate the earnings that would
result from any possible final bundle, so could maximize earnings
without having to manually compute this function.} This function induces payoffs that are symmetric in apples and oranges
and that are strictly variety-seeking. Any subject who prefers more
money to less will wish to maximize this payoff function regardless
of their underlying utility function.

\begin{table}[h]
\begin{tabular}{ccccccccccc}
\toprule 
\addlinespace[3pt]
 &  & \multicolumn{3}{c}{Risk} & \multicolumn{3}{c}{Social} & \multicolumn{3}{c}{Shopping}\tabularnewline\addlinespace[3pt]
\midrule 
\addlinespace[3pt]
Decision & Part & $I$ & Asset A/C & Asset B/D & $I$ & $V_{A}$ & $V_{B}$ & $I$ & $p_{a}$ & $p_{o}$\tabularnewline\addlinespace[3pt]
\midrule
\midrule 
\addlinespace[3pt]
\multirow{2}{*}{D1} & 1 & 10 & $(\$1,\$0)$ & $(\$0,\$1.20)$ & 10 & \$1 & \$1.20 & 8 & 2 & 1\tabularnewline\addlinespace[3pt]
\addlinespace[3pt]
 & 2 & 16 & $(\$1,\$0)$ & $(\$0,\$1)$ & 16 & \$1 & \$1 & 24 & 2 & 2\tabularnewline\addlinespace[3pt]
\midrule 
\addlinespace[3pt]
D2 & 1 & 14 & $(\$2,\$0)$ & $(\$0,\$2)$ & 14 & \$2 & \$2 & 32 & 2 & 1 {\footnotesize{}(for 1st 8),} 2\tabularnewline\addlinespace[3pt]
\midrule 
\addlinespace[3pt]
\multirow{2}{*}{D3} & 1 & 10 & $(\$1,\$0)$ & $(\$0,\$1)$ & 10 & \$1 & \$1 & 30 & 3 & 3\tabularnewline\addlinespace[3pt]
\addlinespace[3pt]
 & 2 & 10 & $(\$1,\$0)$ & $(\$0,\$1.20)$ & 10 & \$1 & \$1.20 & 24 & 3 & 2\tabularnewline\addlinespace[3pt]
\midrule 
\addlinespace[3pt]
D4 & 1 & 16 & $(\$1,\$0)$ & $(\$0,\$1)$ & 16 & \$1 & \$1 & 12 & 1 & 1\tabularnewline\addlinespace[3pt]
\midrule 
\addlinespace[3pt]
D5 & 1 & 10 & $(\$\ensuremath{1},\$0)$ & $(\$0,\$1.20)$ & 10 & \$1 & \$1.20 & 48 & 6 & 4\tabularnewline\addlinespace[3pt]
\midrule 
 &  & \multicolumn{9}{c}{{\footnotesize{}$I$: income for a part (in tokens)}}\tabularnewline
 &  & \multicolumn{3}{c}{{\footnotesize{}$(\$x,\$y)$ indicates one}} & \multicolumn{3}{c}{{\footnotesize{}$V_{A}$: value/token to A}} & \multicolumn{3}{c}{{\footnotesize{}$p_{a}$: price/apple}}\tabularnewline
 &  & \multicolumn{3}{c}{{\footnotesize{}unit of asset pays \$x/\$y}} & \multicolumn{3}{c}{{\footnotesize{}$V_{B}$: value/token to B}} & \multicolumn{3}{c}{{\footnotesize{}$p_{o}$: price/orange}}\tabularnewline
 &  & \multicolumn{3}{c}{{\footnotesize{}if the die roll is 1-3/4-6}} &  &  &  &  &  & \tabularnewline
\end{tabular}

\caption{Experimental Tasks}
\label{tab:design}
\end{table}

Our experimental budget sets, summarized in Table \ref{tab:design},
allow us to conduct our revealed preference tests in the Risk and
Social Experiments, and to conduct analogous tests that make use of
the induced value function in our Shopping Experiment. Throughout,
we refer to part $k$ of decision $t$ as $\text{Dt.k}$, or Dt if
a round has only one part. The order of decisions and parts was varied
across sessions (Table \ref{tab:orders} in Online Supplement \ref{sec:experimental materials}).

We implemented our experiments on paper and followed up with some
robustness treatments online (discussed further in Section \ref{sec:Secondary-analyses}
and Online Supplement \ref{sec:Online-Risk-Experiment}). In each
session, paper instructions were provided and read aloud, subjects
were given the opportunity to ask questions privately, and then participants
completed a brief comprehension quiz and the experimenter individually
checked answers and explained any errors. The comprehension quiz had
each subject calculate how payment would be determined from their
allocations when a decision involves two parts. This was designed
to ensure that each subject was aware of how to calculate payments
in these cases, but without instructing them to consider all possible
combinations of allocations across parts.

Each decision had a cover page indicating the number of parts in the
decision, with each part stapled beneath as a separate page. Thus,
a subject was always informed when a decision contained multiple parts,
but could choose whether or not to look at both parts before making
allocations. A subject indicated each allocation by highlighting the
line corresponding to their choice for each part using a provided
highlighter. Subjects were allowed no other aids at their desk when
making choices. Only one decision was handed out to a subject at a
time, and that decision was collected before the next round was handed
out. The order of decisions, and of parts within each decision, was
varied across sessions to allow us to test and control for order effects.

Sessions took place in Toronto Experimental Economics Lab and SFU
Experimental Economics Lab from June 2019 to February 2020, each taking
place in a one hour block. Subjects were recruited from the labs'
student participant pools to participate in one of the three experiments.
Instructions, experimental materials, and details of the experimental
procedure are provided in Online Supplement \ref{sec:experimental materials}.\footnote{The three experiments were separately pre-registered through the Open
Science Framework. Our analysis follows our plans with only minor
modifications that do not affect the interpretation of our results.
See Appendix \ref{sec:RelationToPrereg} for details.}

\section{Results\label{sec:Results}}

This section reports the results of our experimental tests of the
models of bracketing. For each test, we also compute results allowing
for one or two ``errors'' relative to its requirements. We define
an error as how far we would need to move a subject's allocations
for them to pass that test, measured in lines on the decision sheet(s).\footnote{Example decision sheets are provided in Figures \ref{fig:D11sheet},
\ref{fig:D12}, \ref{fig:D5sheet}, and \ref{fig:D32sheet} in Appendix
\ref{sec:experimental materials}.} For instance, in Risk and Social, a subject is within one error of
passing a test if revising the choice by shifting one token from one
asset/person to the other in a single part would lead to choices that
pass that test. They are within two errors of passing a test if shifting
two tokens from one asset/person to the other, either in the same
part or in different parts, would lead to choices that pass. For predictions
that require Walrasian budget sets, we modify the tests to account
for discreteness in our experiment.

In Section \ref{sec:Secondary-analyses}, we examine the robustness
of the following analysis to two particular concerns. First, we argue
our results are robust to changes in the presentation of the decision.
Second, broad bracketing requires allocating all the budget to one
good in two-part decisions, and earlier work suggests subjects may
be loath to do so. We argue there that neither affects our conclusions
much.

\subsection{Risk Experiment: How subjects behave}

Consider Decision 1 in the Risk Experiment. Since Asset B in the first
part is a perfect substitute for Asset D in the second, comparative
advantage requires that the subject first purchases it in the part
where it has a lower opportunity cost. Consequently, a broad bracketer
necessarily allocates all of their wealth in the first part to Asset
B before they allocate any to Asset D in the second (as required by
BB-Mon). If preferences are symmetric across the two equally-likely
states, then broad bracketing further requires that $x_{A}^{1,1}=0$.
With risk aversion, broad bracketing makes even stronger predictions.
The allocation $x_{A}^{1,1}=0$ and $x_{C}^{1,2}=14$ obtains a risk-free
return of \$14, and any other feasible allocation results in a first-
or second-order stochastically dominated distribution over returns
at the decision-level. In contrast, a narrow bracketer facing Decision
1 allocates at least half of their budget to Asset B in Part 1. If
they are also risk-averse, then they allocate exactly half their budget
in Part 2 to each asset.

\begin{figure}[h]
\caption{Allocations in Risk\label{fig:alloc}}
\subfloat[D1 Allocation in Risk]{\includegraphics[width=0.45\textwidth]{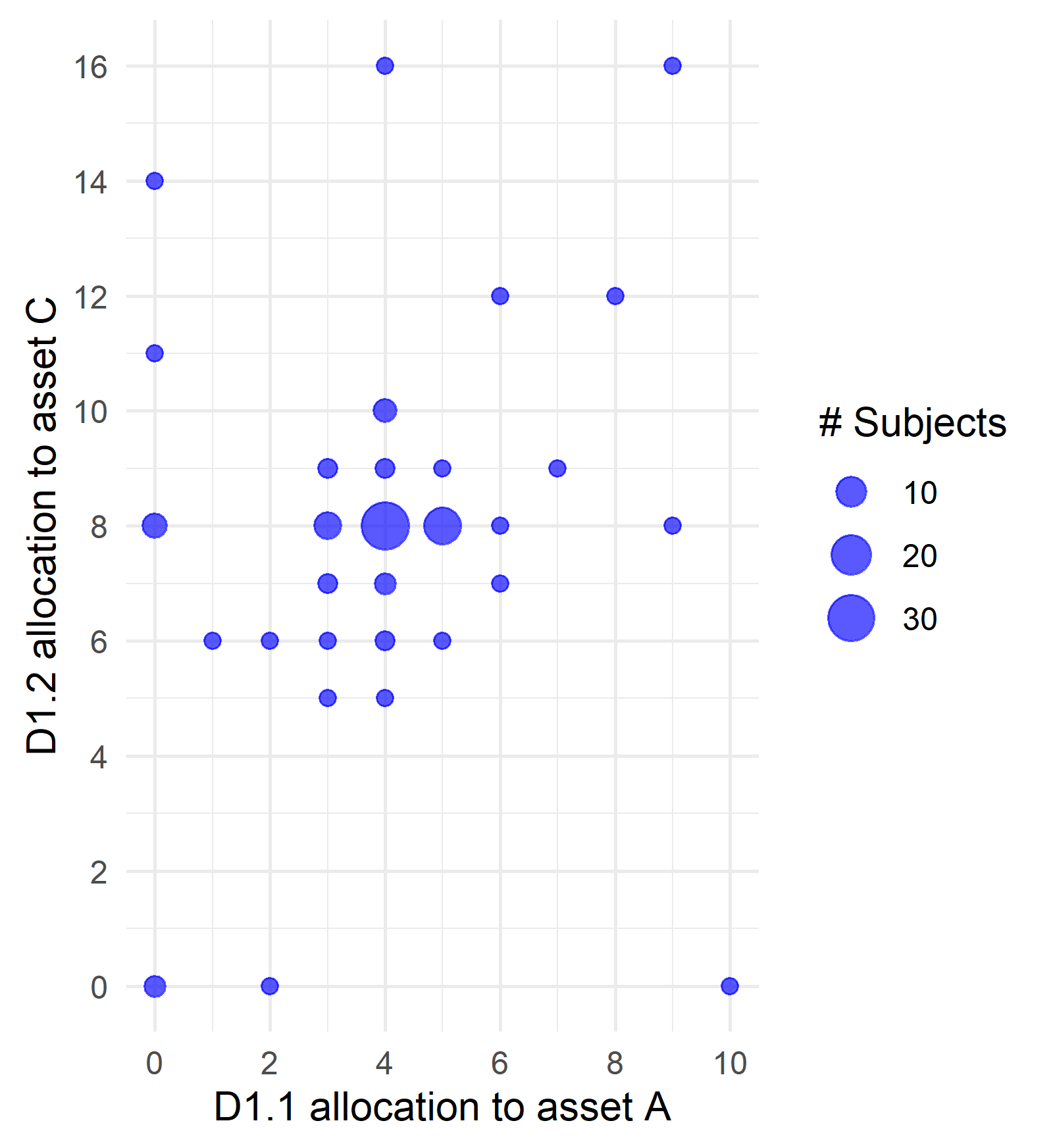}}\subfloat[D3.2 and D5 Allocations]{\includegraphics[width=0.45\textwidth]{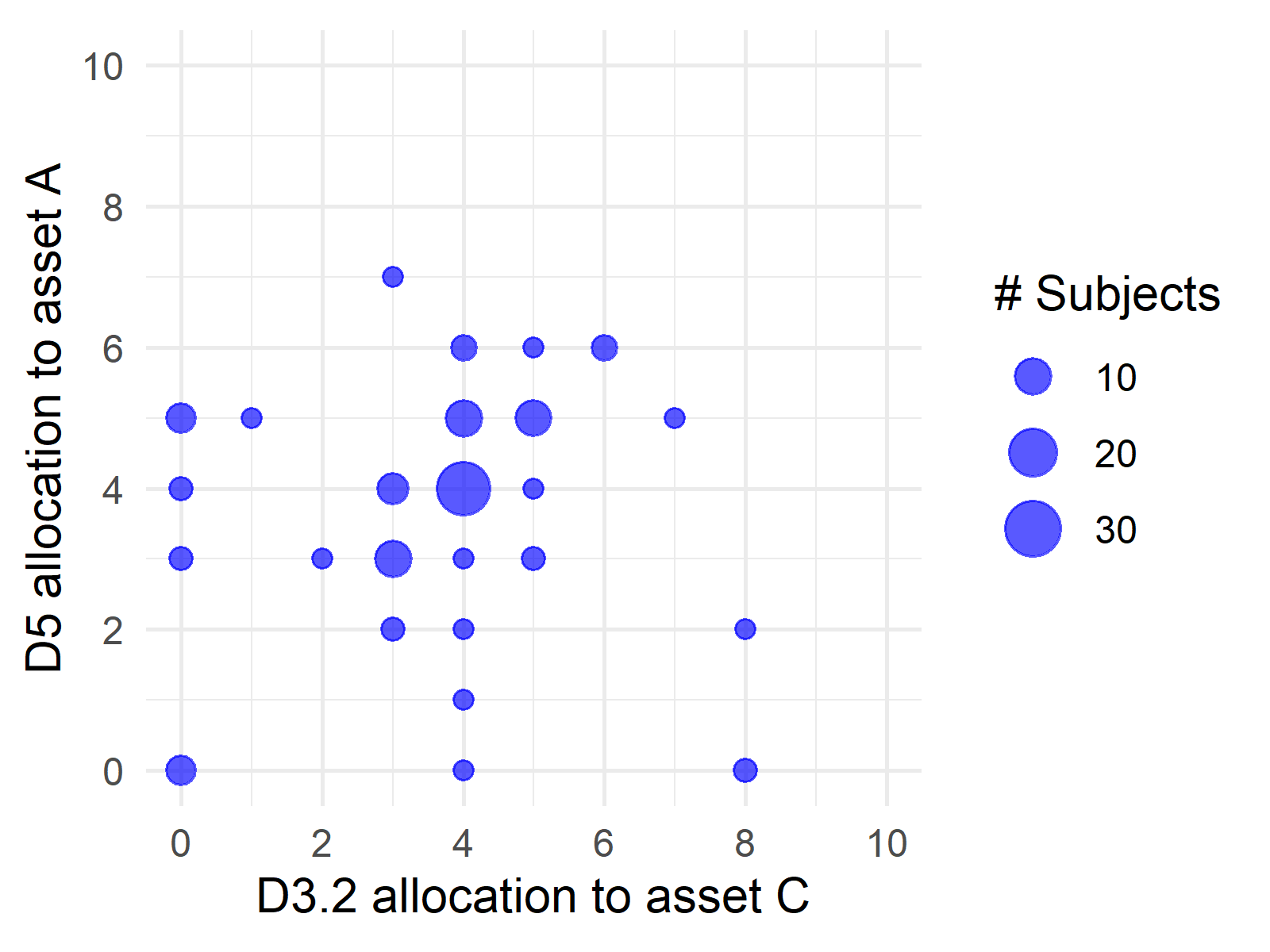}}
\end{figure}

To compare how our subjects perform to this benchmark, Figure \ref{fig:alloc}(A)
plots the joint distribution of their allocations in D1. The x-coordinate
describes their allocation to Asset A in D1.1, and the y-coordinate
their allocation to Asset C in D1.2. The above discussion shows that
broad bracketers will have allocations on the y-axis, and a risk-averse
broad bracketer will select the allocation $\left(x_{A}^{1,1},x_{C}^{1,2}\right)=(0,14)$.
Any narrow bracketer will select an allocation with $x_{A}^{1,1}\leq5$,
and any risk-averse narrow bracketer will select $x_{C}^{1,2}=8$.
The plot shows that few subjects are close to consistent with broad
bracketing, and only a single subject makes an allocation close to
the prediction of risk-averse broad bracketing. Narrow bracketing
with risk aversion is consistent with the modal allocation of $\left(x_{A}^{1,1},x_{C}^{1,2}\right)=\left(4,8\right)$.

In our design, narrow bracketing makes strong predictions across decisions
that have a part in common. For instance, D3.2 and D5 both ask subjects
to allocate 10 tokens between identical assets with identical prices.
A narrow bracketer would make the same allocation in each of the two
parts. To illustrate, Figure \ref{fig:alloc}(B) plots each subject's
allocation in D3.2 against their allocation in D5. The x-coordinate
describes their allocation to Asset C in D3.2, and the y-coordinate
their allocation to its counterpart, Asset A in D5. A narrow bracketer's
choices will fall on the 45-degree line as required by NB-WARP, and
the farther the allocations are from that line the farther the subject
is from narrow bracketing. In the plot, we can see that this prediction
of narrow bracketing holds exactly for 54 of the 99 subjects. We visually
represent all the data underlying our tests in Appendix \ref{sec:Visualizations}.

\subsection{Risk and Social Experiments: Revealed Preference Tests}

We begin by performing the direct revealed preference tests of bracketing
developed in Section \ref{sec:Theory}: NB-WARP, BB-WARP, and BB-Mon
(Table \ref{tab:WARPtests}).

\begin{table}[h]
\begin{tabular}{ccccccccc}
 &  & \multicolumn{3}{c}{Risk} &  & \multicolumn{3}{c}{Social}\tabularnewline
\# errors & $\;$ & 0 & 1 & 2 & $\;$ & 0 & 1 & 2\tabularnewline
\hline 
\hline 
\noalign{\vskip2pt}
NB-WARP (D1.1 and D5) &  & 56 & 76 & 89 &  & 45 & 70 & 77\tabularnewline[2pt]
\noalign{\vskip2pt}
\noalign{\vskip2pt}
NB-WARP (D1.2 and D4)) &  & 56 & 74 & 81 &  & 63 & 78 & 82\tabularnewline[2pt]
\noalign{\vskip2pt}
\noalign{\vskip2pt}
NB-WARP (D3.2 and D5) &  & 54 & 76 & 83 &  & 49 & 75 & 80\tabularnewline[2pt]
\noalign{\vskip2pt}
\noalign{\vskip2pt}
NB-WARP (D1.1 and D3.2) &  & 49 & 76 & 85 &  & 51 & 83 & 87\tabularnewline[2pt]
\noalign{\vskip2pt}
\noalign{\vskip2pt}
NB-WARP (all) &  & 29 & 44 & 61 &  & 28 & 54 & 64\tabularnewline[2pt]
\noalign{\vskip2pt}
\noalign{\vskip2pt}
BB-WARP (D1 and D2) &  & 13 & 20 & 87 &  & 16 & 20 & 94\tabularnewline[2pt]
\noalign{\vskip2pt}
\noalign{\vskip2pt}
BB-Mon (D1) &  & 12 & 13 & 15 &  & 14 & 14 & 14\tabularnewline[2pt]
\noalign{\vskip2pt}
\noalign{\vskip2pt}
BB-Mon (D3) &  & 14 & 16 & 18 &  & 17 & 17 & 18\tabularnewline[2pt]
\noalign{\vskip2pt}
\noalign{\vskip2pt}
BB-Mon (both) &  & 7 & 8 & 10 &  & 12 & 12 & 12\tabularnewline[2pt]
\hline 
\noalign{\vskip2pt}
\noalign{\vskip2pt}
\# subjects &  & \multicolumn{3}{c}{99} &  & \multicolumn{3}{c}{102}\tabularnewline[2pt]
\noalign{\vskip2pt}
\multicolumn{9}{c}{{\scriptsize{}Entries count the \# of subjects who pass test at the
listed error allowance.}}\tabularnewline
\end{tabular}

\caption{Tests of NB-WARP and BB-WARP}
\label{tab:WARPtests}
\end{table}

Very few subjects are consistent with rationality and broad bracketing.
There is only a single pair of decisions (D1 and D2) where choices
could directly violate BB-WARP. For that pair, we test BB-WARP by
comparing the final bundle for D1 to the final bundle for D2: for
any choice of $x^{2,1}$ in D2 with $x_{A}^{2,1}\leq8$, the same
final bundle can be achieved in D1. In each of Risk and Social, only
20\% of subjects are within one error of passing BB-WARP.\footnote{Allowing for two errors substantially raises pass rates of BB-WARP
in Risk and Social. As predicted by both risk averse narrow and broad
bracketing, many subjects, 66\% in Risk and 84\% in Social, choose
$x_{A}^{2,1}=7$. The jump at two errors happens because a subject
who allocates $x_{A}^{2,1}>8$ trivially satisfies BB-WARP.} Even fewer subjects are consistent with BB-Mon than with BB-WARP.
In Risk and Social respectively, we find that 8\% and 12\% of subjects
are within one error of passing BB-Mon in both decisions. Looking
separately at D1 and D3, between 13\% - 17\% of subjects are within
one error of passing BB-Mon.

All told, the BB-WARP and BB-Mon tests provide evidence showing that
80\%-92\% of subjects are not broad bracketers. These rates of violations
of broad bracketing are qualitatively similar to, but higher than,
those found by \citet{tversky1981framing} (73\%) and \citet{rabin2009narrow}
(28\%-66\%), and very close to the structural estimates of the latter
(89\%). In these prior experiments, each part consisted of a pairwise
choice, so failures to broadly bracket are detected only for a particular
range of risk preferences. In contrast, there are many ways a subject
could reveal their failure to bracket broadly in our experiments,
which gives us more power to detect failures.

While previous work can only falsify broad bracketing, our design
allows us to test narrow bracketing as well. We test NB-WARP by comparing
the allocations in each of the two parts that appear in multiple decisions.
Specifically, NB-WARP requires that a subject makes the same choice
in D1.1, D3.2, and D5, and the same choice in D1.2 and D4. Far more
subjects pass each NB-WARP test than either BB-WARP or BB-Mon. Between
75-77\% of subjects in Risk and 69-81\% of subjects in Social are
within one error of passing each of the pairwise NB-WARP tests. Allowing
for one error, 44\% and 53\% of subjects pass all possible NB-WARP
tests in Risk and Social, respectively.\footnote{In contrast, randomly-generated choices have only a 0.4\% chance of
being within one error of passing all NB-WARP tests.}

\begin{result}When allowing for one error, 44\% and 53\% of subjects
are consistent with WARP under narrow bracketing, 20\% and 20\% are
consistent with WARP under broad bracketing, and 8\% and 12\% are
consistent with monotonicity under broad bracketing in the Risk and
Social Experiments, respectively.

\end{result}

We next conduct our revealed preference tests of the three models
considered using the entire set of decisions for each subject. Notice
that in both experiments, the alternative $(x,y)$ leads to an identical
outcome as the alternative $(y,x)$.\footnote{Either leads to a lottery paying $x$ and $y$ with equal probability,
or to a payment of $x$ to one anonymous individual and $y$ to another.} As discussed in Section \ref{sec:Theory}, we extend the tests so
that whenever $(a,b)P^{\mathcal{D}}(x,y)$, i.e., $(a,b)$ is directly
revealed preferred to $(x,y)$ in the dataset $\mathcal{D}$, we also
have $(a,b)P^{\mathcal{D}}(y,x)$, as well as $(b,a)P^{\mathcal{D}}(x,y)$
and $(b,a)P^{\mathcal{D}}(y,x)$. This reduces the need to compare
across decisions and makes our tests more demanding. For example,
all three tests make point predictions in D2 and D4. Random behavior
has a 0.001\% chance or less of passing either BB- or NB-SARP with
one error, and a 0.3\% chance of passing the PNB Algorithm with one
error (see Appendix \ref{sec:Bronars}). Table \ref{tab:FullTests}
shows how many subjects pass each test.

\begin{table}[h]
\begin{tabular}{cc>{\centering}p{0.04\paperwidth}>{\centering}p{0.04\paperwidth}>{\centering}p{0.04\paperwidth}>{\centering}p{0.04\paperwidth}>{\centering}p{0.04\paperwidth}>{\centering}p{0.04\paperwidth}>{\centering}p{0.04\paperwidth}}
\noalign{\vskip2pt}
 &  & \multicolumn{3}{c}{Risk} &  & \multicolumn{3}{c}{Social}\tabularnewline[2pt]
\noalign{\vskip2pt}
\# errors &  & 0 & 1 & 2 &  & 0 & 1 & 2\tabularnewline
\hline 
\noalign{\vskip2pt}
NB-SARP &  & 23 & 34 & 43 &  & 15 & 36 & 44\tabularnewline[2pt]
\noalign{\vskip2pt}
\noalign{\vskip2pt}
BB-SARP &  & 0 & 0 & 0 &  & 8 & 10 & 10\tabularnewline[2pt]
\noalign{\vskip2pt}
\noalign{\vskip2pt}
PNB &  & 49 & 59 & 71 &  & 31 & 58 & 69\tabularnewline[2pt]
\noalign{\vskip2pt}
\noalign{\vskip2pt}
PNB-PE &  & 51 & 63 & 75 &  & 39 & 67 & 76\tabularnewline[2pt]
\noalign{\vskip2pt}
\# subjects &  & \multicolumn{3}{c}{99} &  & \multicolumn{3}{c}{102}\tabularnewline
\multicolumn{9}{c}{{\scriptsize{}Entries count the \# of subjects who pass each test
at the listed error allowance.}}\tabularnewline
\end{tabular}

\caption{Full Tests of Symmetric Models}
\label{tab:FullTests}
\end{table}

We compare results of the tests allowing for up to one error. We find
that no subjects pass BB-SARP in Risk, and only 10\% of subjects pass
it in Social. However, 34-35\% of subjects in each experiment pass
NB-SARP. While fewer subjects pass BB-SARP than either BB-WARP or
BB-Mon, a similar number of subjects pass all NB-WARP restrictions
with one error and the more demanding NB-SARP with two errors. These
results show that a plurality of our subjects are well-described as
narrow bracketers.

Our test of partial-narrow bracketing diagnoses how many of those
who fail BB-SARP and NB-SARP behave consistently with intermediate
degrees of bracketing. We find that 15\% of subjects in Risk and 12\%
of subjects in Social pass the PNB test but neither BB-SARP nor NB-SARP
when allowing for one error.\footnote{Broad and narrow bracketing are both special cases of partial-narrow
bracketing, and thus any subject who passes BB- or NB-SARP will also
pass this test.} Only 4\% of subjects in Risk and 9\% of subjects in Social pass the
PNB-PE test but not the PNB test.

\begin{result}

When allowing for one error relative to each test, 34\% and 35\% of
subjects are rationalizable by narrow bracketing,
0\% and 10\% of subjects pass are rationalizable by broad bracketing in the Risk and Social Experiments, respectively.
An additional 25\% and 12\% of subjects are rationalizable by partial narrow bracketing but not narrow nor broad bracketing.

\end{result}

We note briefly that a subject whose underlying preferences are symmetric
and linear would make the same choices regardless of how they bracket.\footnote{Symmetric linear preferences require only that $x_{A}^{1,1}=x_{A/C}^{3,2}=x_{A}^{5,1}=0$.
Linear preferences are the only class of preferences consistent with
both narrow and broad bracketing (Appendix \ref{sec:Theoretical-Appendix},
Theorem \ref{thm:NB+BB}).} They would be indifferent between many bundles in our experiment
and so violate our tests' assumption that choices reveal strict preferences.
However, this cannot bias our conclusions too much since only five
subjects in Risk and two subjects in Social could possibly be consistent
with it.\footnote{One subject in each of Risk and Social pass NB-SARP, and the other
subject in Social passes BB-SARP when allowing one error. Two other
subjects in Risk also had $x_{C}^{1,2}=x_{A}^{3,1}=0$, and so pass
BB-WARP, BB-Mon, and NB-WARP but neither NB-SARP nor BB-SARP. In total,
one in each of Risk and Social are classified as broad, two in Risk
and one in Social are classified as narrow, and the remaining two
in Risk are unclassified.}

\subsection{Shopping Experiment: Induced-Value Tests}

The tests of our predictions thus far assume that utility is not observed.
When utility is known, as in our Shopping Experiment, narrow and broad
bracketing each make unique predictions in each decision.\footnote{Narrow-bracketed maximization implies $x_{a}^{1,1}=1$, $x_{a}^{1,2}=6$,
$x_{a}^{3,1}=5,$ and $x_{a}^{3,2}=4$. Broad-bracketed maximization
implies $x_{a}^{1,1}=0,x_{a}^{1,2}=10,x_{a}^{3,1}=10,$ and $x_{a}^{3,2}=0$.} To test the models, we compare how far each subject's choices are
from each model's predictions (Table \ref{tab:ShoppingTests}).

\begin{table}[h]
\begin{tabular}{ccccccccccccccccc}
\noalign{\vskip2pt}
 &  & \multicolumn{3}{c}{D1} &  & \multicolumn{3}{c}{D3} &  & \multicolumn{3}{c}{Both} &  & \multicolumn{3}{c}{Full}\tabularnewline[2pt]
\noalign{\vskip2pt}
\# errors & $\;$ & 0 & 1 & 2 & $\;$ & 0 & 1 & 2 & $\;$ & 0 & 1 & 2 & $\;$ & 0 & 1 & 2\tabularnewline
\hline 
\noalign{\vskip2pt}
NB &  & 23 & 23 & 60 &  & 53 & 65 & 69 &  & 20 & 21 & 49 &  & 15 & 16 & 40\tabularnewline[2pt]
\noalign{\vskip2pt}
\noalign{\vskip2pt}
BB &  & 21 & 22 & 27 &  & 23 & 24 & 25 &  & 13 & 15 & 18 &  & 12 & 14 & 16\tabularnewline[2pt]
\noalign{\vskip2pt}
\noalign{\vskip2pt}
PNB &  & 45 & 46 & 95 &  & 76 & 91 & 98 &  & 33 & 37 & 68 &  & 27 & 30 & 56\tabularnewline[2pt]
\noalign{\vskip2pt}
\noalign{\vskip2pt}
PNB-PE &  & 45 & 46 & 95 &  & 76 & 93 & 98 &  & 33 & 37 & 70 &  & 27 & 30 & 57\tabularnewline[2pt]
\noalign{\vskip2pt}
\# subjects &  & \multicolumn{15}{c}{101}\tabularnewline
\multicolumn{17}{c}{{\scriptsize{}Entries count the \# of subjects who pass each test
at the listed error allowance.}}\tabularnewline
\end{tabular}

\caption{Shopping Tests}
\label{tab:ShoppingTests}
\end{table}

Testing the point predictions of narrow bracketing in each of Decisions
1 and 3, 23\% and 64\% of subjects are respectively within one error
of the predictions of narrow-bracketing, while 21\% are consistent
in both. Allowing for two errors raises pass rates to 59\% for Decision
1 and 49\% for both.\footnote{The sharp difference between one- and two-error tests in Decision
1 but not Decision 3 results from the fact that 40 subjects selected
$x_{a}^{1,1}=2,x_{o}^{1,1}=4$, which is the second-best available
bundle from a narrow bracketer's perspective but is two lines away
from the best bundle of $x_{a}^{1,1}=1,x_{o}^{1,1}=6$ -- due to
the discreteness of the budget set, one apple and five oranges is
between these bundles. We note that this deviation is in the opposite
direction of that predicted by broad or partial-narrow bracketing,
and 32 of these subjects also selected $x_{a}^{1,2}=x_{o}^{1,2}=6$.} Looking at the full set of implications of narrow-bracketing on all
choices made in the experiment, 40\% of subjects are within two errors
of passing.

In contrast, 21\%, 24\%, and 14\% of subjects are within one error
of being consistent with broadly-bracketed maximization in Decisions
1, 3, and both, respectively. When allowing for two errors, those
numbers remain similar. Using all decisions in the experiment, only
16\% of subjects are within two errors of being consistent with all
implications of broad-bracketed maximization. PNB describes less than
2\% of the remaining subjects.

Since we induce the payoff function, we are able to compute the value
of $\alpha$ that best fits a subjects behavior (up to the limits
imposed by discretization of the budget sets) under the assumption
that the induced value function acts as their utility function. To
that end, we compute the point predictions of the partial-narrow bracketing
model for each $\alpha\in\{0,0.01,0.02,\dots,0.99,1\}$ for Decisions
1 and 3, and obtain distinct predictions for nine intervals of $\alpha$.
We assign each subject to the range of $\alpha$ for which their choices
exhibit the fewest errors relative to that range's predictions. We
find that 64\% of subjects are classified to a range that includes
full narrow bracketing, $\alpha=1$, and 25\% are classified to a
range that includes full broad bracketing, $\alpha=0$. Of the remaining
subjects, none are best described by $\alpha\in[0.25,0.71]$. This
suggests that even those subjects who are not exactly described by
either broad or narrow bracketing are close.\footnote{Because of discreteness, $[0,1]$ can be partitioned so that the parameters
in each cell make the same choices. In this case, the partition is
$\{[0,0.04],[0.05,0.06],[0.07,0.14],[0.15,0.24],[0.25,0.29],[0.3,0.38],[0.39,0.71],[0.72,0.75],[0.76,1]\}$.}

\subsection{Classifying subjects to models}

The tests thus far do not make any adjustment for the fact that partial-narrow
bracketing nests narrow and broad bracketing as polar cases, and can
thus accommodate more behavior. To compare the predictive success
of each model at the subject level, we use a subject-level implementation
of the Selten score \citep{selten1991properties,beatty2011demanding}.
For each subject and each model (symmetric versions for Risk and Social,
using the induced value function for Shopping), we calculate the number
of errors the subject exhibits relative to that model. Then, we calculate
the number of possible choice combinations in the experiment that
are consistent with that model and that number of errors -- the model-error
pair's ``predictive area''. We divide the predictive area by the
total number of possible combinations of choices in the experiment
to compute the measure for each subject $i$ and model $m\in\text{\{broad, narrow, PNB},\text{ PNB-PE}\}$
as $\text{predictive\_success}_{i,m}=1-\frac{\#\text{predictive\_area\_for\_}i,m}{\#\text{all\_possible\_choices}}$.\footnote{In the Risk and Social Experiments, there are $(11\times17)\times15\times(11\times11)\times17\times11=63,468,735$
possible combinations of choices. Symmetric narrow bracketing allows
6, 87, and 606 possible combinations of choices when allowing for
zero, one, and two errors, respectively, whereas symmetric broad bracketing
allows 12, 116, and 585 combinations of choices, symmetric PNB allows
35,797, 200,828, and 597,728 combinations, and symmetric PNB-PE allows
116,267, 619,375, and 1,725,466 combinations. Thus, a subject whose
choices are consistent with partial-narrow bracketing will be classified
as a partial-narrow bracketer if and only if they are sufficiently
far from being consistent with both broad and narrow bracketing.} We use all choices made in the experiment to assign each subject
to the model with the highest predictive success; in cases where every
rationalizing model-error pair for a subject would rationalize more
than one million possible combinations of choices in our experiment,
we categorize them as ``Unclassified''.
\begin{table}[h]
\begin{tabular}{cccc}
 & \multicolumn{3}{c}{Percent Selten Score Maximized}\tabularnewline
\noalign{\vskip2pt}
 & Risk & Social & Shopping\tabularnewline[2pt]
\hline 
\noalign{\vskip2pt}
\hline 
\noalign{\vskip2pt}
Broad Bracketing & 2.02 & 9.80 & 26.73\tabularnewline[2pt]
\noalign{\vskip2pt}
\noalign{\vskip2pt}
 & {\footnotesize{}(0.35,7.81)} & {\footnotesize{}(4.32,15.29)} & {\footnotesize{}(17.99,35.48)}\tabularnewline[2pt]
\noalign{\vskip2pt}
\noalign{\vskip2pt}
Narrow Bracketing & 77.78 & 75.49 & 67.32\tabularnewline[2pt]
\noalign{\vskip2pt}
\noalign{\vskip2pt}
 & {\footnotesize{}(69.49,86.06)} & {\footnotesize{}(67.10,83.87)} & {\footnotesize{}(57.98,76.67)}\tabularnewline[2pt]
\noalign{\vskip2pt}
\noalign{\vskip2pt}
PNB & 7.07 & 1.96 & 3.96\tabularnewline[2pt]
\noalign{\vskip2pt}
\noalign{\vskip2pt}
 & {\footnotesize{}(3.13,14.51)} & {\footnotesize{}(0,4.68)} & {\footnotesize{}(0.13,7.78)}\tabularnewline[2pt]
\noalign{\vskip2pt}
\noalign{\vskip2pt}
PNB-PE & 0 & 3.92 & 0.99\tabularnewline[2pt]
\noalign{\vskip2pt}
\noalign{\vskip2pt}
 & {\footnotesize{}(0,0)} & {\footnotesize{}(0.22,7.61)} & {\footnotesize{}(0,2.84)}\tabularnewline[2pt]
\noalign{\vskip2pt}
\noalign{\vskip2pt}
Unclassified & 13.13 & 8.82 & 0.99\tabularnewline[2pt]
\noalign{\vskip2pt}
\noalign{\vskip2pt}
 & {\footnotesize{}(6.42,19.83)} & {\footnotesize{}(3.46,14.18)} & {\footnotesize{}(0,2.91)}\tabularnewline[2pt]
\noalign{\vskip2pt}
\noalign{\vskip2pt}
\multicolumn{4}{c}{{\scriptsize{}To calculate confidence intervals, we assume that the
model of bracketing is multinomially}}\tabularnewline[2pt]
\noalign{\vskip2pt}
\noalign{\vskip2pt}
\multicolumn{4}{c}{{\scriptsize{}distributed with fixed strike rate within each treatment
and calculate the \citet{wilson1927probable} score.}}\tabularnewline[2pt]
\noalign{\vskip2pt}
\end{tabular}

\caption{Classification of subjects}
\label{tab:classification}
\end{table}

Across the three experiments, we classify 67-78\% of subjects as narrow
bracketers (Table \ref{tab:classification}). In contrast, 2\%, 10\%,
and 27\% of subjects are classified as broad bracketers in the Risk,
Social, and Shopping Experiments respectively. However, 7\%, 6\%,
and 5\% were classified to one of the two partial-narrow bracketing
models. After adjusting for predictive power, partial-narrow bracketing
does not help explain very many subjects' behavior.

To form confidence intervals, we assume that model of bracketing is
multinomially distributed with fixed strike rate within each treatment.
We calculate confidence intervals according to the \citet{Wilson1927}
score. In all three treatments, the lower bound of the interval for
narrow bracketers exceeds the upper bound of the intervals for all
other models of bracketing. Most other intervals overlap, but there
are significantly more broad bracketers than partial-narrow bracketers
in the Shopping Experiment.

\begin{result}Judging each model's fit by its predictive success,
67-78\% of subjects are classified as narrow bracketers, 2-27\% are
classified as broad bracketers, and 5-7\% are classified as partial-narrow
bracketers across the three experiments.

\end{result}

\section{Robustness and Secondary Analyses\label{sec:Secondary-analyses}}

To study how choice architecture mediates bracketing, we conducted
an online version of our Risk Experiment that varied the presentation
in a two-by-two design. First, the ``Examine'' treatment instructed
the subject to ``First, examine both accounts, then purchase your
investments'' in the instructions, tested this in a quiz question,
then included that text at the top of each two-part decision screen;
the ``Basic'' treatment did not. Second, the ``Tabs'' treatment
presented each part of a decision as a separate HTML tab (analogous
to our separate pages in our paper experiments), whereas the ``Side-by-Side''
treatment presented parts of a decision side-by-side on the same screen
(analogous to \citealt{tversky1981framing}). We recruited 200 US-based
subjects from Prolific Academic to participate and randomly assigned
each to one of the four treatment pairs.

We also replicated our Shopping experiment online with 46 subjects
from Prolific Academic. For all 46, we implemented the Side-by-Side
and Examine interventions, eliminated the quantity-restricted sale
in D2, and provided subjects with a calculator instead of a payoff
table. We provide detailed screenshots and results in the Online Supplement.

In Section \ref{subsec:choice architecture}, we use the online experiments
to argue that presentation has a limited effect on our results.The
online experiments allow us to collect non-choice data that shed light
on the subjects' choice processes. In Section \ref{subsec:Welfare-losses},
we use this data to argue that a substantial fraction of narrow bracketers
give some consideration to both parts of the decision before choosing.
We then explore the robustness of our tests of broad bracketing to
extremeness aversion, investigate the possibility of order effects
or learning, and perform statistical tests of the differences we find
across the domains.

\subsection{Effects of choice architecture\label{subsec:choice architecture}}

The online subject pool is more slanted towards narrow bracketing
than our original Pen-and-paper study. Pass rates for each NB-WARP
test are higher, and pass rates for BB-WARP and BB-Mon are generally
lower. Only 2 of 200 (1\%) subjects pass BB-SARP when allowing for
two errors, and these two, and only these two, are classified as broad
bracketers according to Selten score. In contrast, 96 (48\%) subjects
pass NB-SARP when allowing one error, and 84.5\% are classified as
narrow bracketers by Selten score.

We find almost no effect on choices from either of the two interventions
in the Online Risk Experiment. Both subjects who pass BB-SARP are
in the Examine and Tabs treatment, contrary to our expectation that
Side-by-Side would be more conducive to broad bracketing. Neither
treatment has a large or statistically significant (with a small sample
size caveat) effect on the rate of broad bracketing ($p=0.22,p=0.24$,
for Fisher's exact tests of the Examine and Tabs treatments, respectively).
Within the Tabs group, Examine does not have a statistically significant
effect ($p=0.20$, Fisher's exact test). We similarly find no effect
of the treatments on the rates of narrow bracketing ($p=1.00$ for
both, Fisher's exact tests).

In addition, the Online Shopping Experiment implemented both Examine
and Side-by-Side. It provided subjects with a calculator instead of
a payoff table. We classified 5 out of 46 (10.86\%) subjects to broad
bracketing and 33 (71.74\%) to narrow. The rate of narrow bracketing
did not change much, while the rate of broad bracketing slightly decreased
relative to the pen-and-paper Shopping Experiment ($p=0.03$, Fisher's
exact test).

All-in-all, this suggests that the low rates of broad bracketing we
find are not overly sensitive to varying the choice architecture to
encourage broad bracketing. We caution against over-interpreting this
result since broad bracketing was so rare in all treatments. The shift
to an online interface and the Prolific subject pool may have had
more effect than any of the nudges. We suspect that more extreme nudges
and decision aids might be more effective.

\subsection{Non-choice data and bracketing\label{subsec:Non-choice-data}}

Our online experiments shed light on how the choice process, and in
particular consideration, differed across subjects. We have two tools
for measuring consideration. In the Tabs arm of Online Risk, we observe
which tabs subjects clicked and when they made their decision. In
Online Shopping, we record how subjects used the calculator. This
non-choice data shows that only some narrow bracketers completely
ignore other parts of the decision. In both experiments, more than
a quarter of narrow bracketers gave some consideration to both parts
of the decision. These subjects had enough information to bracket
more broadly yet did not.

In the Tabs versions of the Online Risk Experiment, we record whether
each subject clicked on both tabs before making their final choices.
Only 28 of 102 subjects did so in both D1 and D3. While this includes
the 2 broad bracketers, the other 26 subjects had sufficient information
to bracket broadly but failed to do so. Of the 28, 22 were in the
Examine-Tabs treatment, and 6 were in the Basic-Tabs treatment ($p<0.01$,
Fisher's exact). The prompt was effective at increasing this click
pattern, but it did not significantly affect the classification to
broad bracketing ($p=0.20$, Fisher's exact). Clearly, paying attention
to both parts of a decision is necessary to bracket broadly. However,
the fraction of narrow bracketers among those who paid attention was
not substantially different than the fraction among those who did
not (21 of 28 versus 61 of 74, $p=0.41$, Fisher's exact). At the
other extreme, we find that 33 subjects made their final choice in
each of D1 and D3 before having seen both parts, and thus could not
have been broad bracketing; 29 of these subjects are classified as
narrow bracketers.

While intensity of calculator use does not differ much between broad
and narrow bracketers in the Online Shopping Experiment, patterns
of calculator use do.\footnote{The average (median) number of calculations was 10.40 (12) for broad
bracketers, 12.21 (8) for narrow bracketers, and 12.86 (10) for all
subjects. As a whole,} Ex-ante, we expected that plugging in bundles that are feasible in
the decision as a whole (e.g. (9, 11) for D3) to the calculator (a
``broad calculation'') would predict broad bracketing, and that
plugging in a bundle that is available in a part of a decision but
lies below the decision-level feasible set (e.g. (5,5) for D3) to
the calculator (a ``narrow calculation'') would predict narrow bracketing.
We classify nine of ten subjects of who made a narrow calculation
as narrow bracketers. Surprisingly, we classify nine of 16 subjects
who made at least one such ``broad calculation'' as narrow bracketers,
and only three as broad. These subjects appear to actively attend
to how the parts fit together, yet still make narrowly-bracketed choices.

In summary, the non-choice data suggest diversity in the relationship
between consideration and bracketing. It suggests that about a quarter
of narrow bracketers appear to follow a choice process that involves
only considering one part at a time, but this same data also shows
that a quarter or more of narrow bracketers gave some consideration
to both parts. The seemingly weak link between broader consideration
and broadly-bracketed choices suggests why the interventions we designed
to encourage broad bracketing in the Online Risk Experiment were not
very effective. However, the sparsity of broad bracketers in both
experiments makes it hard for us to tell what drives it. The choice
processes leading to narrow or broad bracketing is a fertile direction
for future work.

\subsection{Effects of Extremeness Aversion\label{subsec:Extremeness}}

As we describe in Section \ref{sec:Theory}, broad bracketing requires
a corner solution in at least one part of any multi-part decision.
Evidence from other contexts suggests that some subjects may be \emph{extremeness
averse} and avoid corner choices. For example, subjects in linear
public good games tend not to play the Nash strategy of making no
contributions. However, in non-linear public good games where the
Nash strategy requires a positive but not complete contribution, subjects
play the equilibrium strategy more frequently (see e.g. Section 2
of \citealp{vesterlund2016using}, for a discussion). Extremeness
aversion could affect our conclusions about broad bracketing and its
prevalence relative to narrow bracketing.

Non-extreme allocations in two-part decisions lead to violations of
BB-Mon. However, few additional subjects are consistent with a relaxation
of BB-Mon that allows for subjects to be close to, but not necessarily
at, a corner when required. Formally, we say that a subject passes
\emph{Extremeness-Averse (EA-)BB-Mon} if they are within 2 tokens
of making an extreme allocations in each decision required by BB-Mon.
Without allowing for any errors, 3 subjects in Risk and 0 subjects
in Social pass EA-BB-Mon but not BB-Mon. Allowing for 2 errors, 14\%
of subjects in Risk and 3\% in Social pass EA-BB-Mon but not BB-Mon.
As a benchmark, 26\% of all possible datasets pass EA-BB-Mon but not
BB-Mon when allowing for two errors.

Extremeness aversion is unlikely to affect our classification or tests.
Of the 77 subjects classified as narrow bracketers in each of Risk
and Social, none pass EA-BB-Mon but not BB-Mon (allowing for two errors,
this increases to six and one, respectively). No subject in either
Risk or Social passes both NB-SARP and EA-BB-Mon but not BB-Mon when
allowing for no or one errors in each test. Allowing two errors in
each test, only 3 subjects in Risk and 1 in Social pass both NB-SARP
and EA-BB-Mon but not BB-Mon. Since BB-Mon is a necessary but not
sufficient condition for broad bracketing, a highly risk-tolerant
or inequity-neutral narrow bracketer would pass both NB-SARP and BB-Mon.
Only one subject in each does so. Somewhat more risk- or inequity-averse
narrow bracketers will also pass EA-BB-Mon. 

This suggests that the effect of extremeness aversion on our tests
is probably not too large. In addition, a similar number of subjects
are consistent with broad bracketing in decisions that require two
corner choices as are consistent in decisions that require only one
corner choice.\footnote{A broad bracketer in Shopping buys only apples in D3.1, only oranges
in D3.2, and only oranges in D1.1, but buys both fruits in D1.2. D3
requires two extreme choices, and D1 requires only one. As reported
in Table \ref{tab:ShoppingTests}, there are a very similar number
of subjects close to these predictions in D1 and D3.} The evidence still suggests heterogeneity in bracketing, with a plurality
best described as narrow bracketers, even after adjusting for extremeness
aversion.\footnote{We also conducted a follow-up experiment for which only one of the
two treatments requires a corner choice. To get more separation between
narrow and broad, we used a much finer choice grid for each budget.
Unfortunately, this made the data too noisy to draw strong conclusions.
No more than 14.8\% of subjects are close to the predictions of either
narrow or broad bracketing combined in either treatment. Nonetheless,
the ratio of broad to narrow bracketers does not vary much across
treatments, lending weak support to our conclusion. We elaborate on
this experiment in the Online Supplement, Section \ref{sec:Online-Shopping-Experiment}.}

\subsection{Welfare losses and measurement\label{subsec:Welfare-losses}}

So far, we have focused on identifying how our subjects bracketed.
We now apply our classification to quantify in dollar terms the payoffs
forgone by narrow bracketers (Table \ref{tab:losses}). We also illustrate
how looking at aggregate data can lead to misleading inferences. A
narrow bracketer's preference parameters should be measured using
parts as the unit of observation, and broad bracketer's should use
decisions as the unit of observation. Using the wrong unit of observation,
necessary for an aggregate analysis with heterogeneous bracketing,
leads to misleading conclusions about preferences. We focus on the
Social and Shopping Experiments because there are only two participants
classified as broad bracketers in Risk.

\begin{table}[h]
\begin{tabular}{cc>{\centering}p{0.05\paperwidth}>{\centering}p{0.05\paperwidth}>{\centering}p{0.05\paperwidth}>{\centering}p{0.05\paperwidth}>{\centering}p{0.05\paperwidth}>{\centering}p{0.05\paperwidth}>{\centering}p{0.05\paperwidth}>{\centering}p{0.05\paperwidth}}
\multicolumn{2}{c}{Panel A} &  &  &  &  &  &  &  & \tabularnewline
 &  & \multicolumn{7}{c}{Payoff loss (\$)} & \tabularnewline
 &  & D1.1 & D1.2 & D1 & D3.1 & D3.2 & D3 & D5 & $n$\tabularnewline
\hline 
\multicolumn{2}{c}{Risk} & 0.39 & 0 & 0.39 & 0 & 0.35 & 0.35 & 0.38 & 99\tabularnewline
 & NB & 0.40 & 0 & 0.40 & 0 & 0.38 & 0.38 & 0.41 & 77\tabularnewline
 & BB & 0 & 0 & 0 & 0 & 0.20 & 0.20 & 0.30 & 2\tabularnewline
\multicolumn{2}{c}{Social} & 0.47 & 0 & 0.47 & 0 & 0.46 & 0.46 & 0.49 & 102\tabularnewline
 & NB & 0.51 & 0 & 0.51 & 0 & 0.51 & 0.51 & 0.49 & 77\tabularnewline
 & BB & 0 & 0 & 0 & 0 & \multicolumn{1}{c}{0.03} & 0.03 & 0.44 & 10\tabularnewline
\multicolumn{2}{c}{Shopping} & 0.56 & 0.51 & 1.06 & 1.03 & 0.85 & 1.25 & 0.05 & 101\tabularnewline
 & NB & 0.22 & 0.07 & 1.41 & 0.01 & 0.07 & 1.67 & 0.03 & 68\tabularnewline
 & BB & 1.23 & 1.31 & 0.35 & 3.76 & 2.94 & 0.14 & 0.04 & 27\tabularnewline
\multicolumn{10}{l}{{\scriptsize{}``Payoff loss'' columns report total \$ loss (Shopping),
loss in expected value (Risk), and}}\tabularnewline
\multicolumn{10}{l}{{\scriptsize{}average loss per recipient (Social) compared to expected/average
value maximization. We}}\tabularnewline
\multicolumn{10}{l}{{\scriptsize{}evaluate each part on its own in columns D1.1 and D3.2
and integrate across parts in}}\tabularnewline
\multicolumn{10}{l}{{\scriptsize{}columns D1 and D3. In D1.2 and D3.1 all allocations
have no loss in either Risk or Social.}}\tabularnewline
 &  &  &  &  &  &  &  &  & \tabularnewline
\end{tabular}

\begin{tabular}{cc>{\centering}p{0.05\paperwidth}>{\centering}p{0.05\paperwidth}>{\centering}p{0.05\paperwidth}>{\centering}p{0.05\paperwidth}>{\centering}p{0.05\paperwidth}>{\centering}p{0.05\paperwidth}>{\centering}p{0.05\paperwidth}>{\centering}p{0.05\paperwidth}}
\multicolumn{2}{c}{Panel B} &  &  &  &  &  &  &  & \tabularnewline
 &  & \multicolumn{7}{c}{Payoff difference (\$)} & \tabularnewline
 &  & D1.1 & D1.2 & D1 & D3.1 & D3.2 & D3 & D5 & $n$\tabularnewline
\hline 
\hline 
\multicolumn{2}{c}{Risk} & 4.24 & 2.24 & 5.39 & 1.88 & 4.70 & 4.51 & 3.79 & 99\tabularnewline
 & NB & 3.47 & 0.91 & 3.91 & 0.70 & 3.72 & 3.98 & 3.24 & 77\tabularnewline
 & BB & 12.00 & 9.00 & 3.00 & 3.00 & 7.60 & 4.60 & 6.60 & 2\tabularnewline
\multicolumn{2}{c}{Social} & 3.10 & 2.39 & 1.54 & 2.29 & 3.35 & 1.91 & 2.58 & 102\tabularnewline
 & NB & 1.64 & 0.51 & 1.62 & 0.60 & 1.68 & 1.68 & 2.17 & 77\tabularnewline
 & BB & 12.00 & 12.00 & 0.00 & 10.00 & 11.34 & 1.46 & 2.56 & 10\tabularnewline
\multicolumn{10}{c}{{\scriptsize{}``Payoff difference'' columns report the difference
in \$ allocations across the two states or people.}}\tabularnewline
\end{tabular}\caption{Losses and payoff differences by parts and decisions}
\label{tab:losses}
\end{table}

In the Shopping Experiment, we can precisely measure welfare losses
through the earnings that participants did not receive.\footnote{Interpreting payoff losses as welfare losses requires that broad bracketing
does not have a direct utility cost.} Narrow bracketers had lower earnings than broad bracketers in both
two-part decisions. They had \$1.41 and \$1.67 in forgone payments
compared to just \$0.35 and \$0.14. The differences between the maximum
and the minimum payoffs are \$9.60 and \$10.36 in D1 and D3, so the
payoff losses for narrow bracketers are 10.9\% and 14.8\% of the variable
payment. In the one-part decision D5, narrow and broad bracketers
should make identical choices. They do indeed make very similar choices,
suggesting that arithmetic errors do not explain the difference. Indeed,
narrow bracketers do a good job at optimizing part-by-part, and would
forgo no more than \$0.22 if the compensation was based on individual
parts of a decision rather than the decision as a whole.\footnote{There is a similar pattern for expected or aggregate payoffs in the
Risk and Social Experiments, but risk or inequity aversion make interpreting
this as a welfare loss more ambiguous.}

The ``Payoff difference'' columns of Table \ref{tab:losses} illustrate
how it is easy to obtain misleading inferences about how subjects
make the equity-efficiency trade-off. In the one-part decision D5,
narrow bracketers implemented allocations that differed by an average
of \$2.17 between A and B, while the corresponding average difference
is \$2.56 for broad bracketers ($p=0.18$, rank sum test comparing
D5 allocations). Narrow and broad bracketers make similar trade-offs
between equity and efficiency from one-part decisions. Yet if we look
at D1.2 or D3.1 on its own, ignoring the other part of the decision,
we would incorrectly infer that broad bracketers care mainly about
efficiency, while narrow bracketers make near-equal allocations on
average, consistent with inequity aversion. We would obtain an opposite
conclusion by looking at D1 as a whole, where narrow bracketers' decision-level
allocations involve some inequity on average, but all broad bracketers
make perfectly equal final allocations.

\subsection{Differences across experiments}

The fraction of subjects classified to broad bracketing varies across
experiments, from 1\% for Online Risk to 27\% in Pen-and-paper Shopping.
There is a significant difference in broad bracketing rates between
Pen-and-paper Shopping and each of Risk and Social ($p<0.01$, Fisher's
exact test) as well as between the Online and Pen-and-paper Shopping
experiments ($p<0.01)$.\footnote{The difference between Online and Pen-and-paper Risk was not significant,
but this is to be expected given the negative effect on broad bracketing
and the already small numbers of broad bracketers.} However, narrow bracketing rates varied much less, from 67\% to 78\%
across Pen-and-paper experiments, with no significant differences
in these rates ($p>0.10$ for all pairwise Fisher's exact tests).
Nor was there a significant difference in narrow bracketing rates
between the Online and Pen-and-paper versions of Risk and Shopping
($p>0.10$ for both Fisher's exact tests). Explanations for the higher
rate of broad bracketing in Shopping and lower rate in Risk include
the more naturalistic setting, the presence of an objectively-correct
payoff function, and cognitive difficulties specific to choice under
risk (as suggested by \citealt{martinez2019failures}). We cannot
distinguish between these explanations.

\section{Conclusion}

We propose revealed preference tests for how a person brackets their
choices that rely only on monotonicity of underlying preferences.
We deploy these tests in an experiment where both narrow and bracketing
make falsifiable predictions, unlike in past work. Across our experiments,
at least twice as many subjects were classified as narrow bracketers
than as broad bracketers. A majority of people tend to narrowly bracket,
while a noticeable minority broadly bracket. While many of our subjects
are not well-described by either broad or narrow bracketing, our novel
tests of partial-narrow bracketing suggest that it does not do much
better after adjusting for predictive power. This suggests that applications
should calibrate a population mix of broad and narrow bracketers rather
than a representative agent model with a calibrated partial-narrow
bracketing parameter (as in \citealp{BarberisHuang2007equity}).

Bracketing rates differ across tasks, domains, and subject pools.
While our framework is well-suited to detect and to measure these
differences, it is less well-suited to determine why these differences
persist. Non-choice data, as we collected in the online follow up
experiment, can help. It seems to rule out uniform explanations for
why so many bracket narrowly, such as lack of awareness of complementarities
across parts. We think understanding why people bracket the way they
do is an interesting direction for future work.

\bibliographystyle{authordate1}
\bibliography{bracketing}

\appendix

\section{Theoretical Appendix: Proofs, Derivations, and PNB Algorithm\label{sec:Theoretical-Appendix}}

\subsubsection*{Proof of Theorem \ref{thm:The-following-are}.}

One can follow the usual proof that SARP holds if and only if a dataset
is rationalizable by a complete and transitive preference relation
to establish the result with a preference relation; see e.g. \citealp{mas-colell.whinston.ea95}.
Establishing rationalization with a utility function requires a bit
more work. We provide an argument below as we are unaware of one in
the literature.

Given $\mathcal{O}=\left(x^{i},B^{i}\right)_{i=1}^{N}$ for either
$\mathcal{O}=\mathcal{D}^{NB}$ or $\mathcal{O}=\mathcal{D}^{BB}$,
set $X^{\prime}=\cup_{i}B^{i}$. Observe that each $B^{i}$ is a compact
set, either by assumption (for $\mathcal{O}=\mathcal{D}^{NB}$) or
since the finite sum of compact sets is compact (for $\mathcal{O}=\mathcal{D}^{BB}$).
Thus $X^{\prime}$ is also compact. It is therefore bounded as a subset
of $\mathbb{R}^{n}$. Let $X$ be a closed ball centered at $0$ containing
$X^{\prime}$.

Let $\succsim$ be the transitive closure of $P^{\mathcal{O}}$ defined
on $X$. We show that $\succsim$ is closed as a subset of $X\times X$.
First, note $P^{\mathcal{O}}$ itself is closed as a finite union
of the closed sets $\{(y,z)\in X\times X:y\geq z\}$ and $\{x^{i}\}\times B^{i}$
for all $i=1,\dots,N$. Second, we show that $y\succsim z$ if and
only if there exist $y_{1},y_{2},\dots,y_{K-1},y_{K}$ so that $y_{i}P^{\mathcal{O}}y_{i+1}$
for all $i$, $y_{1}=y$, and $y_{K}=z$ where $K\leq2N+1$. Consider
any $y\succsim z$. By definition, there exist $y_{1},y_{2},\dots,y_{K}$
so that $y_{i}P^{\mathcal{O}}y_{i+1}$ for all $i$, $y_{1}=y$, and
$y_{K}=z$ and there is no shorter sequence that establishes $y\succsim z$.
If $y_{j}\geq y_{j+1}$ and $y_{j+1}\geq y_{j+2}$ for some $i$,
then $y_{j}\geq y_{j+2}$, and we can construct a shorter sequence
by leaving $y_{j+1}$ out. If $y_{j}\not\geq y_{j+1}$, then $y_{j}=x^{i(j)}$
and $y_{j+1}\in B^{i(j)}$ for some $i(j)$. Then, $y_{k}\neq y_{j}$
for all $k>j$ by SARP, so at most one index $j^{\prime}$ has $i(j^{\prime})=i$
for each $i$. Therefore, the longest possible shortest sequence alternates
$\geq$ with direct revelations from choice. There can be no more
than $2N+1$ of these, so $K\leq2N+1$. Moreover, since we can trivially
append $y_{K+1}=y_{K}$ to the end of such sequences, we can consider
only sequences with exactly $2N+1$ in what follows.

Now, pick any sequence $(y_{m},z_{m})_{m=1}^{\infty}$ contained in
$\succsim$ that converges to $(y,z)$. Then, there exists $y_{m}^{2},\dots,y_{m}^{2N}$
so that $y_{m}P^{\mathcal{O}}y_{m}^{2}P^{\mathcal{O}}\dots P^{\mathcal{O}}y_{m}^{2N}P^{\mathcal{O}}z_{m}$.
The sequence $(y_{m},y_{m}^{2},\dots,y_{m}^{2N},z_{m})$ has a convergent
subsequence $(y_{m_{k}},y_{m_{k}}^{2},\dots,y_{m_{k}}^{2N},z_{m_{k}})$
since $X\times\dots\times X$ is compact by the Tychonoff Theorem.
Let $(y^{*},y^{2*},\dots,y^{2N*},z^{*})$ be its limit. Since $P^{\mathcal{O}}$
is closed, $y^{*}P^{\mathcal{O}}y^{2*}P^{\mathcal{O}}\dots P^{\mathcal{O}}y^{2N*}P^{\mathcal{O}}z^{*}$.
Finally, $\left(y^{*},z^{*}\right)=\left(y,z\right)$ since $\left(y_{m},z_{m}\right)\rightarrow\left(y,z\right)$.
Hence $\succsim$ is closed.

Now, apply \citealp{Levin1983}'s Theorem, as in \citealp{NishimuraOkQuah2017},
to get a utility function $U$ from $X$ to $[0,1]$ so that $x\succ y$
implies $U(x)>U(y)$ and $x\succsim y$ implies $U(x)\geq U(y)$.
Since $\geq$ is included in $\succsim$, $U(\cdot)$ is increasing.
Since $x^{i}\succ y$ for all $y\in B^{i}\setminus\left\{ x^{i}\right\} $,
$U(x^{i})>U(y)$ for all $y\in B^{i}\setminus\left\{ x^{i}\right\} $.
Extending to all of $\mathbb{R}_{+}^{n}$ is straightforward, simply
set $U(y)=V(y)$ whenever $y\notin X$ for an arbitrary increasing
function $V(\cdot)$ picked so that $V(y)>1$ for all $y$.

\subsubsection*{Proof of Theorem\ref{thm: alpha PNB}.}

Fix $\alpha$ and $\mathcal{D}$. For every decision $t$, define
the lottery
\[
p^{t}=\frac{\alpha K_{t}}{\alpha K_{t}+(1-\alpha)}\left(\frac{1}{K_{t}},x^{t,k}\right)_{k=1}^{K_{t}}+\frac{(1-\alpha)}{\alpha K_{t}+(1-\alpha)}\left(1,x^{t}\right),
\]
and $Y\subseteq\mathbb{R}^{n}$ to be a finite set that includes the
union of the supports of $p^{1},\dots,p^{T}$ and the vector $0$.
It will be convenient to take $Y$ equal to this set but this is inessential.
Denote by $\Delta Y$ the finite support lotteries over $Y$. Consider
the set 
\[
Q_{Y}^{t}=\left\{ \frac{\alpha K_{t}}{\alpha K_{t}+(1-\alpha)}\left(\frac{1}{K_{t}},y^{k}\right)_{k=1}^{K_{t}}+\frac{(1-\alpha)}{\alpha K_{t}+(1-\alpha)}\left(1,y\right)\in\Delta Y:\exists z^{k}\in B^{t,k}\text{ so that }y^{k}\leq z^{k}\,\&\ y\leq\sum z^{k}\right\} 
\]
is a finite set of lotteries over $Y$ that includes $p^{t}$ and
any others that are affordable. The ancillary dataset 
\[
\mathcal{D}_{Y}^{\alpha}=\left\{ \left(p^{t},Q_{Y}^{t}\right)\right\} _{t=1}^{T}
\]
  is strictly rationalizable by expected utility if and only if $\mathcal{D}$
is rationalizable by $\alpha$-partial-narrow bracketing.

If so, then there exists an EU preference $V$ over $\Delta Y$ so
that
\[
V\left(p^{t}\right)>V\left(q^{t}\right)\,\forall q^{t}\in Q_{Y}^{t}\backslash p^{t}
\]
 for all $t$. Let $v$ be the utility index of $V$ and extend to
$\mathbb{R}_{+}^{n}$ by 
\[
u^{*}(z)=\max_{y\leq z,y\in Y}v(y).
\]
Then for any $t$ and $y^{t,1}\times\dots\times y^{t,K_{t}}\in B^{t,1}\times\dots\times B^{t,K_{t}}\setminus\left\{ x^{t,1}\times\dots\times x^{t,K_{t}}\right\} $,
\[
\alpha\sum_{k}u^{*}(x^{t,k})+(1-\alpha)u^{*}(x^{t})=\kappa_{t}V\left(p^{t}\right)>\kappa_{t}\max_{q^{t}\in Q_{Y}^{t}\backslash p^{t}}V\left(q^{t}\right)\geq\alpha\sum_{k}u^{*}(y^{t,k})+(1-\alpha)u^{*}\left(\sum_{k=1}^{K_{t}}y^{t,k}\right)
\]
where $\kappa_{t}=\alpha K_{t}+(1-\alpha)$.\footnote{The max exists since $Q_{Y}^{t}$ is a finite set.}
To find a strictly increasing $u$, let $u^{**}$ be a strictly increasing
extension of $v$. Choosing $\epsilon>0$ small enough, for $u=(1-\epsilon)u^{*}+\epsilon u^{**}$,
\[
\alpha\sum_{k}u(x^{t,k})+(1-\alpha)u(x^{t})>\alpha\sum_{k}u(y^{t,k})+(1-\alpha)u(\sum_{k=1}^{K_{t}}y^{t,k})
\]
 for all $t$ and any $y^{t,1}\times\dots\times y^{t,K_{t}}\in B^{t,1}\times\dots\times B^{t,K_{t}}\setminus\left\{ x^{t,1}\times\dots\times x^{t,K_{t}}\right\} $,
establishing that $\mathcal{D}$ is rationalized by $\alpha$-PNB.\footnote{This approach is similar to that of \citet{Polissonetal2020}.}

Conditions under which $\mathcal{D}_{Y}^{\alpha}$ is rationalized
by expected utility are well-known. Here, we follow \citet{clark1993revealed}.
As standard, $p$ is directly revealed preferred to $q$, written
$p\succ q$, if $p=p^{t}$ and $q\in Q_{Y}^{t}\setminus\left\{ p^{t}\right\} $
for some decision $t$. Define $p\succsim q$ if $p\succ q$ or $p=q$.
We say that $p$ is indirectly reveal preferred via independence to
$q$, written $p\tilde{\succsim}q$, whenever there exist $r,p_{1},\dots,p_{n},q_{1},\dots,q_{n}\in\Delta Y$,
$\beta\in(0,1]$, and $\lambda_{1},\dots,\lambda_{n}>0$ so that 
\begin{align*}
\beta p+(1-\beta)r & =\sum_{i=1}^{n}\lambda_{i}p_{i}\\
\beta q+(1-\beta)r & =\sum_{i=1}^{n}\lambda_{i}q_{i}\\
p_{i} & \succsim q_{i}\text{ for all }i=1,\dots,n\\
\sum_{i=1}^{n}\lambda_{i} & =1
\end{align*}
and $p\tilde{\succ}q$ holds whenever at $p_{i}\succ q_{i}$ for some
$i$. The Linear Axiom of Revealed Preference (LARP) is that $q\tilde{\succsim}p$
implies that $p\tilde{\succ}q$ does not hold. Theorem 3 of \citet{clark1993revealed},
combined with the finiteness of $\mathcal{D}^{\alpha}$, implies that
$\tilde{\succsim}$ has an expected utility rationalization. Conclude
$\mathcal{D}^{\alpha}$ satisfies LARP if and only if $\mathcal{D}$
is rationalized by $\alpha$-partial-narrow bracketing.

The following algorithm inputs $\mathcal{D}^{\alpha}$ and outputs
$1$ if and only $\mathcal{D}$ is rationalized by $\alpha$-partial-narrow
bracketing for any rational number $\alpha\in[0,1]$. For concreteness,
label $Y=\left\{ y_{1},y_{2},\dots,y_{m}\right\} $. We identify each
$p\in\Delta Y$ with the $p\in\mathbb{R}^{m}$ so that $p_{i}=p(y_{i})$.

Define the matrix $A$ with a column equal to $p^{t}-q$ for each
$q\in Q_{Y}^{t}\backslash p^{t}$ and every $t$. $A$ is an $D\times m$
matrix interpreted as strict preference, where $D$ is the number
of data points extracted. Let $\vec{1}$ be a $D\times1$ matrix of
ones and 
\[
\hat{A}=\left[\begin{array}{c}
A\\
\vec{1}
\end{array}\right].
\]
Let $\hat{b}$ be an $m+1$ vector with the first $m$ components
$0$ and the last $1$.

Note that if $p\tilde{\succsim}q$, then there exists $\beta>0,$$\lambda_{i}\geq0$,
and $p_{i}\succsim q_{i}$ for $i=1,...,m$ so that 
\[
\beta(p-q)=\sum_{i=1}^{m}\lambda_{i}(p_{i}-q_{i})
\]
and $\sum_{i=1}^{m}\lambda_{i}=1$. Each $p_{i}-q_{i}$ corresponds
to a column of $A$, so this is equivalent to $p-q\in C=cone(co(\left\{ A_{1},\dots,A_{D}\right\} ))$
where $A_{i}$ is the $i$th column of $A$. LARP holds when $p-q\in C$
implies $q-p\notin C$, or, equivalently, if and only if $0\notin C$.
That is, there does not exist $\phi^{1},\dots,\phi^{D}\geq0$ so that
$\sum\phi^{i}=1$ so that $\sum_{i=1}^{D}\phi^{i}A_{i}=0$, which
is equivalent to $\hat{A}\phi=\hat{b}$.\footnote{To extend to allow for some weak preference, e.g. implications of
symmetry, let $B$ be an $D'\times m$ matrix with \textbf{columns}
given by $p-q$ where $p$ and $q$ represent comparisons including
weak preference and $\vec{0}$ be a $D'\times1$ matrix of zeroes.
Repeat the above with $\hat{A}=\left[\begin{array}{cc}
A & B\\
\vec{1} & \vec{0}
\end{array}\right]$.}

Solve the system
\begin{align*}
 & \min_{y\in\mathbb{R}^{m+1}}y^{T}\hat{b}\\
s.t. & \hat{A}^{T}y\geq0
\end{align*}
for $y$. The Ellipsoid algorithm provides a polynomial-time solution
to this. Let $y^{*}$ be the solution, if one exists. There are two
cases:

\textbf{Case 1: }There is no solution to $\hat{A}^{T}y\geq0$ or $y^{*}$
satisfies $y^{*T}\hat{b}\geq0$. Then, Farkas's Lemma implies there
exists $\phi\geq0$ so that $\hat{A}\phi=\hat{b}$. LARP fails by
the above,  so return $0$.

\textbf{Case 2:} The program is unbounded below or $y^{*T}\hat{b}<0$.
Then, Farkas's Lemma implies there \textbf{does not} exist any $\phi\geq0$
so that $\hat{A}\phi=\hat{b}$. LARP is satisfied by the above, so
return 1. 

To specialize the above to PNB-PE, we perform the analysis part by
part rather than decision by decision. Let $p^{t,k}=(\alpha,x^{t};(1-\alpha),x^{t,k})$,
$Y\subseteq\mathbb{R}^{n}$ be a finite set that includes the union
of the supports of the lotteries $p^{t,k}$, and 
\[
Q_{Y}^{t,k}=\left\{ \left(\alpha,y;(1-\alpha),y^{\prime}\right)\in\Delta Y:\text{\ensuremath{\exists}}x\in B^{t,k}\text{ so that }y^{\prime}\leq x+x^{t,-k}\,\&\,y\leq x\right\} 
\]
Repeating the above with the ancillary dataset 
\[
\mathcal{D}_{Y}^{PE,\alpha}=\left\{ \left(p^{t,k},Q^{t,k}\right)\right\} _{(t,k)}
\]
replacing $\mathcal{D}^{\alpha}$ establishes the result.

\hfill{}$\square$
\begin{thm}[Identification of $\alpha$.]
\label{prop:IDalpha} Suppose that $u$ is a known differentiable
and strictly quasi-concave function and $\mathcal{D}$ is rationalized
by partial-narrow bracketing. Further suppose that for some $t,k$,
$B^{t,k}$is a Walrasian budget set for prices $p^{t,k}$, and $n=2$,
$\frac{p_{1}^{t,1}}{p_{2}^{t,1}}\neq\frac{p_{1}^{t,2}}{p_{2}^{t,2}}$,
$\underset{x_{1}\rightarrow0^{+}}{\lim}\frac{\partial u(x_{1},x_{2})}{\partial x_{1}}=+\infty$
for all $x_{2}>0$, and $\underset{x_{2}\rightarrow0^{+}}{\lim}\frac{\partial u(x_{1},x_{2})}{\partial x_{2}}=+\infty$
for all $x_{1}>0$. \\
Then, there exists a unique $\alpha$ that rationalizes choices.
\end{thm}

\subsubsection*{Proof of Theorem \ref{prop:IDalpha}.}

Suppose the assumptions of the Theorem are satisfied. In the partial-narrow
bracketing model, the marginal utility per dollar of spending on good
$j$ in decision $k$ is given by:

\[
\frac{1}{p_{j}^{t,k}}\left((1-\alpha)\frac{\partial u(x^{t,1}+x^{t,2})}{\partial x_{j}}+\alpha\frac{\partial u(x^{t,k})}{\partial x_{j}}\right)
\]
for $k=1,2$. By the limit conditions, if $\alpha>0$, then $x_{j}^{t,k}>0$
for $j=1,2$ and $k=1,2$; thus $x_{1}^{t,k}=0$ or $x_{2}^{t,k}=0$
implies $\alpha=0$. If $\frac{1}{p_{2}^{t,k}}\partial u(x^{t,k})/\partial x_{2}=\frac{1}{p_{1}^{t,k}}\partial u(x^{t,k})/\partial x_{1}$
for both $k$, then set $\alpha=1$. Now suppose that $\frac{1}{p_{2}^{t,k}}\partial u(x^{t,k})/\partial x_{2}\neq\frac{1}{p_{1}^{t,k}}\partial u(x^{t,k})/\partial x_{1}$
and an interior allocation is chosen in budget set $k$. The first-order
condition for an interior maximizer equates the marginal utility per
dollar across the two goods; rearranging this condition, we obtain

\[
\alpha=\frac{\frac{1}{p_{1}^{t,k}}\partial u(x^{t,1}+x^{t,2})/\partial x_{1}-\frac{1}{p_{2}^{t,k}}\partial u(x^{t,1}+x^{t,2})/\partial x_{2}}{\frac{1}{p_{1}^{t,k}}\partial u(x^{t,1}+x^{t,2})/\partial x_{1}-\frac{1}{p_{2}^{t,k}}\partial u(x^{t,1}+x^{t,2})/\partial x_{2}+\frac{1}{p_{2}^{t,k}}\partial u(x^{t,k})/\partial x_{2}-\frac{1}{p_{1}^{t,k}}\partial u(x^{t,k})/\partial x_{1}}
\]
for $k=1,2$. It remains to verify that the expression on the right-hand
side always lies in the interval $(0,1)$. Since $\frac{1}{p_{2}^{t,k}}\partial u(x^{t,k})/\partial x_{2}\neq\frac{1}{p_{1}^{t,k}}\partial u(x^{t,k})/\partial x_{1}$,
the solution satisfies $(1-\alpha)\left(\frac{1}{p_{1}^{t,k}}\frac{\partial u(x^{t,1}+x^{t,2})}{\partial x_{1}}-\frac{1}{p_{2}^{t,k}}\frac{\partial u(x^{t,1}+x^{t,2})}{\partial x_{2}}\right)=\alpha\left(\frac{1}{p_{2}^{t,k}}\frac{\partial u(x^{t,k})}{\partial x_{2}}-\frac{1}{p_{1}^{t,k}}\frac{\partial u(x^{t,k})}{\partial x_{1}}\right)$.
Thus, if $\alpha\in(0,1)$, the denominator and numerator have the
same sign, and the denominator is strictly larger in absolute value
-- therefore the expression is well-defined and thus $\alpha$ is
identified from choices.

\hfill{}$\square$

\subsubsection*{Joint implications of narrow and broad bracketing}

To study the joint implications of broad and narrow bracketing we
consider endowment bracketing in a rich domain. We thus suppose that
we observe an infinite dataset on all decisions that each consist
of a binary part and a second degenerate part summarized by its only
available bundle $z$. Suppose that when a person has bundle $z\in\mathbb{R}_{+}^{n}$
in their second part and must make a binary choice, they apply a complete,
transitive, continuous, and strictly increasing binary relation $\succsim_{z}$
(call such a binary relation a \emph{well-behaved preference}). In
this domain, we say that a person \emph{broadly brackets} if $x\succsim_{z}y\iff x+z\succsim_{\text{\ensuremath{\underbar{0}}}}y+z$;
a person \emph{narrowly brackets} if $x\succsim_{z}y\iff x\succsim_{z^{\prime}}y$
for all $z^{\prime}\in\mathbb{R}_{+}^{n}$.
\begin{thm}
\label{thm:NB+BB} The family of well-behaved preferences $\{\succsim_{z}\}_{z\in\mathbb{R}_{+}^{n}}$
satisfies both narrow and broad bracketing if and only if there exists
a vector $u\in\mathbb{R}_{++}^{n}$ such that $x\succsim_{z}y$ if
and only if $x\cdot u\geq y\cdot u$.
\end{thm}
\begin{proof}
The ``if'' direction is immediate. So now consider the ``only if''
direction.We do so by constructing an additive utility representation
for $\succsim_{\text{\ensuremath{\underbar{0}}}}$. Since $\succsim_{\text{\ensuremath{\underbar{0}}}}$
is a well-behaved preference relation, for each $x\in\mathbb{R}_{+}^{n}$
there exists a unique number $\kappa\in\mathbb{R}_{+}$ such that
$x\sim_{\text{\ensuremath{\underbar{0}}}}\kappa e$ where $e=(1,1,\dots,1)$;
define $u(x)=\kappa$ for each $x$ and the corresponding $\kappa$.
By strict monotonicity, $u$ represents $\succsim_{\text{\ensuremath{\underbar{0}}}}$,
and by continuity, $u$ is continuous. Pick any $x,y\in\mathbb{R}_{+}^{n}$.
By NB, $x\sim_{y}u(x)e$ and $y\sim_{u(x)e}u(y)e$, so by BB and transitivity
$x+y\sim_{0}y+u(x)e\sim_{0}u(x)e+u(y)e$, so $u(x+y)=u(x)+u(y)$.
Since $x,y$ were arbitrary, $u(x)=u\cdot x$ for some $u\in\mathbb{R}_{++}^{n}$
by Theorem 1 in Chapter 5 of \citet{Acz1966}.
\end{proof}

\subsubsection*{Derivations of predictions and their experimental implementations.}

First, observe that if relation $P$ is acyclic, it must be anti-symmetric.

\paragraph*{Derivation of Prediction 1 (NB-WARP)}

Suppose $x^{t^{\prime},k^{\prime}}\in B^{t,k}\subseteq B^{t^{\prime},k^{\prime}}$.
By the definition of $P$ in NB-SARP, $x^{t^{\prime},k^{\prime}}Px$
for all $x\in B^{t^{\prime},k^{\prime}}\backslash\{x^{t^{\prime},k^{\prime}}\}$.
Since $B^{t,k}\subseteq B^{t^{\prime},k^{\prime}}$, it follows that
$x^{t^{\prime},k^{\prime}}Px$ for all $x\in B^{t,k}\backslash\{x^{t^{\prime},k^{\prime}}\}$.
But if $x^{t,k}\neq x^{t^{\prime},k^{\prime}}$, the definition of
$P$ would imply $x^{t,k}Px^{t^{\prime},k^{\prime}}$, which would
violate acyclicity. Thus $x^{t,k}=x^{t^{\prime},k^{\prime}}$.

\subparagraph*{Discretized implementation in our experiment.}

In all of our tests of NB-WARP, we have $B^{t,k}=B^{t^{\prime},k^{\prime}}$
exactly. We thus test NB-WARP exactly, without needing to adjust for
discreteness.

\paragraph*{Derivation of Prediction 2 (BB-WARP)}

Suppose $x^{t^{\prime}}\in B^{t}\subseteq B^{t^{\prime}}$. By the
definition of $P$ in BB-SARP, $x^{t^{\prime},k^{\prime}}Px$ for
all $x\in B^{t^{\prime}}\backslash\{x^{t^{\prime}}\}$. Since $B^{t}\subseteq B^{t^{\prime}}$,
it follows that $x^{t^{\prime}}Px$ for all $x\in B^{t}\backslash\{x^{t^{\prime}}\}$.
But if $x^{t}\neq x^{t^{\prime}}$, the definition of $P$ would imply
$x^{t}Px^{t^{\prime}}$, which would violate acyclicity. Thus $x^{t}=x^{t^{\prime}}$.

\subparagraph*{Discretized implementation in our experiment.}

In the Risk and Social experiments, $\text{co}B^{1,1}+\text{co}B^{1,2}\subseteq\text{co}B^{2,1}$.
In particular, $B^{2,1}=\left\{ \left(\$0,\$28\right);\left(\$2,\$26\right);\dots;\left(\$28,\$0\right)\}\right\} $,
whereas $B^{1,1}+B^{1,2}=\left\{ \left(\$0,\$28\right);\left(\$1,\$27\right);\dots;\left(\$16,\$12\right);\left(\$17,\$10.80\right);\dots;\left(\$26,\$0\right)\right\} $.
Consider two cases.

If $x_{A}^{2,1}\leq\$16$, the bundle chosen in $B^{2,1}$ is exactly
affordable in $B^{1,1}+B^{1,2}$. The caveat is that the set $B^{1,1}+B^{1,2}$
has a higher resolution in this range. We do not adjust for this,
implicitly interpreting their choice from $B^{2,1}$ as revealing
their preferences over the $\text{co}B^{2,1}$. However, if this leads
a subject to fail BB-WARP, if their preferences are well-behaved,
they would choose in $B^{2,1}$ one of the closest bundles to their
preferred bundle in $\text{co}B^{2,1}$ and be within one error away
from passing BB-WARP.

However, if $x_{A}^{2,1}>\$16$, then the person reveals a sufficiently
strong preference for person/state A over B that their desired bundle
in $B^{2}$ is not affordable in $\text{co}B^{1,1}+\text{co}B^{1,2}$
and thus they trivially pass BB-WARP.

\paragraph*{Derivation of Prediction 3 (BB-Mon)}

Suppose $x_{1}^{t,2}>0$, $x_{2}^{t,1}>0$, and $p_{1}^{t,1}\leq p_{2}^{t,1}$,
$p_{1}^{t,2}>p_{2}^{t,2}$. Let $\epsilon=\min\{x_{1}^{t,2},x_{2}^{t,1}\}$.
Consider the alternative pair of choices $\left(y^{t,1},y^{t,2}\right)$
given by $y^{t,1}=\left(x_{1}^{t,1}+\frac{p_{2}^{t,1}}{p_{1}^{t,1}}\epsilon,x_{2}^{t,1}-\epsilon\right)$
and $y^{t,2}=\left(x_{1}^{t,2}-\epsilon,x_{2}^{t,2}+\frac{p_{1}^{t,2}}{p_{2}^{t,2}}\epsilon\right)$.
By construction, $\left(y^{t,1},y^{t,2}\right)$ is affordable. We
have that

$y_{1}^{t,1}+y_{1}^{t,2}=x_{1}^{t,1}+x_{1}^{t,2}+\frac{p_{2}^{t,1}-p_{1}^{t,1}}{p_{1}^{t,1}}\epsilon\geq x_{1}^{t,1}+x_{1}^{t,2}$,
and

$y_{2}^{t,1}+y_{2}^{t,2}=x_{2}^{t,1}+x_{2}^{t,2}+\frac{p_{1}^{t,2}-p_{2}^{t,2}}{p_{2}^{t,1}}\epsilon>x_{2}^{t,1}+x_{2}^{t,2}$

where the inequalities respectively follow from $p_{1}^{t,1}\leq p_{2}^{t,1}$
and $p_{1}^{t,2}>p_{2}^{t,2}$. But $x^{t}Py^{t}$ would violate monotonicity;
thus such an $x^{t,1},x^{t,2}$ pair could not pass BB-SARP.

\subparagraph*{Discretized implementation in our experiment.}

The argument in the proof applies with minimal modification to our
discretized experiment, and thus we apply the conditions directly.\pagebreak{}
\end{document}